%% file: paper.tex
\newif\ifarxiv
\arxivtrue

\documentclass[conference,10pt,compsocconf]{IEEEtran}
\usepackage{amsmath,amsthm}
\usepackage[american]{babel}
\usepackage[inline]{enumitem}
\usepackage{xspace}
\usepackage{graphicx}
\usepackage[detect-all]{siunitx}
\usepackage[nameinlink,noabbrev,capitalize]{cleveref}
\usepackage{algorithm}
\usepackage[noend]{algpseudocode}

\usepackage{caption}
\usepackage{subcaption}
\usepackage{epstopdf}
\graphicspath{{./figures/}}

\newtheorem{definition}{Definition}
\newtheorem{remark}{Remark}

\setlist[enumerate,1]{label=(\arabic*)}
\setlist[enumerate,2]{label=(\alph*),ref=(\arabic{enumi}\alph*)}

\newcommand{\dft}[1]{\emph{#1}}      % for defined terms

\newcommand{\calG}{\mathcal{G}}
\newcommand{\calL}{\mathcal{L}}

\newcommand{\Gmax}{\mathcal{G}_{\mathrm{max}}}

\newcommand{\N}{\mathbf{N}}          % natural
\newcommand{\Zpos}{\mathbf{Z}}          % positive integer
\newcommand{\R}{\mathbf{R}}          % reals
\newcommand{\nnR}{\R_{\geq 0}}        % non-neg reals

\newcommand{\e}{\varepsilon}

\newcommand{\abs}[1]{\left|#1\right|}
\newcommand{\paren}[1]{\left(#1\right)}
\newcommand{\set}[1]{\left\{#1\right\}}
\newcommand{\der}[1]{\frac{\mbox{\footnotesize \textup{d}}}{\mbox{\footnotesize \textup{d#1}}}}	% derivatives

\newcommand{\OGCS}{\textsf{OffsetGCS}}
\newcommand{\FT}{\textbf{FT}}
\newcommand{\offset}{\widehat{O}}
\newcommand{\omax}{\offset_{\mathrm{max}}}
\newcommand{\omin}{\offset_{\mathrm{min}}}
\newcommand{\Tmax}{T_{\mathrm{max}}}
\newcommand{\Tosc}{T_{\mathrm{osc}}}
\newcommand{\Tctr}{T_{\mathrm{cnt}}}
\newcommand{\Tmeas}{T_{\mathrm{meas}}}
\newcommand{\metas}{\texttt{M}}
\newcommand{\modes}{\textsc{mode}}

\newtheorem{lemma}{Lemma}
\newtheorem{corollary}{Corollary}
\newtheorem{theorem}{Theorem}

\usepackage{color}
\usepackage{soul}

\begin{document}

\title{{\huge PALS: Plesiochronous and Locally Synchronous Systems}}

\author{\IEEEauthorblockN{%
Johannes Bund\IEEEauthorrefmark{1}\IEEEauthorrefmark{2},
Matthias F{\"u}gger\IEEEauthorrefmark{3},
Christoph Lenzen\IEEEauthorrefmark{1}, 
Moti Medina\IEEEauthorrefmark{4} and
Will Rosenbaum\IEEEauthorrefmark{1}}
\IEEEauthorblockA{\IEEEauthorrefmark{1}%
MPI for Informatics,
Saarland Informatics Campus\\
\{jbund, clenzen, wrosenba\}@mpi-inf.mpg.de}
\IEEEauthorblockA{\IEEEauthorrefmark{2}%
Saarbr\"ucken Graduate School of Computer Science}
\IEEEauthorblockA{\IEEEauthorrefmark{3}%
CNRS \& LSV, ENS Paris-Saclay, Universit\'e Paris-Saclay \& Inria\\
mfuegger@lsv.fr}
\IEEEauthorblockA{\IEEEauthorrefmark{3}%
Ben-Gurion University of the Negev\\
medinamo@bgu.ac.il}%
}

\maketitle

\begin{abstract}
Consider an arbitrary network of communicating modules on a chip, each requiring
a local signal telling it when to execute a computational step. There are three
common solutions to generating such a local clock signal:
  (i) by deriving it from a single, central clock source,
  (ii) by local, free-running oscillators, or
  (iii) by handshaking between neighboring modules.

Conceptually, each of these solutions is the result of a perceived dichotomy in which
(sub)systems are either clocked or fully asynchronous, suggesting that the
designer's choice is limited to deciding where to draw the line between
synchronous and asynchronous design.

In contrast, we take the view that the better question to ask is \emph{how}
synchronous the system can and should be. Based on a distributed clock
synchronization algorithm, we present a~novel~design providing modules with
local clocks whose frequency bounds are almost as good as those of corresponding
free-running oscillators, yet neighboring modules are guaranteed to have a phase
offset substantially smaller than one clock cycle.
Concretely, parameters obtained from a $15\,\text{nm}$ ASIC implementation
running at $2$\,GHz yield mathematical worst-case bounds of $30$\,ps on phase offset
for a $32\times 32$ node grid network.
\end{abstract}


\input{intro}
\input{algo}
\input{modules}

\input{implementation}
\input{simulation}
\input{conclusion}

\smallskip
\noindent{\bf Acknowledgments.}
We thank the reviewers for their valuable feedback, and in particular the
third reviewer for pointers to related work.
This research has received funding from the European Research Council (ERC)
under the European Union’s Horizon 2020 research and innovation programme (grant agreement No 716562),
the Israel Science Foundation under Grant 867/19, ANR grant
FREDDA (ANR-17-CE40-0013), and the Digicosme working group HicDiesMeus.

\bibliographystyle{IEEEtran}
\bibliography{bibliography}

\ifarxiv % ------ arxiv begin
\appendices
\input{proofs}
\fi % ------- arxiv end

\end{document}

%% file: intro.tex
\section{Introduction and Related Work}\label{sec:intro}
At surface level, the synchronous and asynchronous design paradigms seem to be
opposing extremes. In their most pure forms, this is true: Early synchronous
systems would wait for a clock signal to be propagated throughout the system and
all computations of the current clock cycle to complete before moving on to the
next; and delay-insensitive circuits make no assumptions on timing whatsoever,
explicitly acknowleding completion of any computational step.

In reality, however, fully synchronous or asynchronous systems are the
exception. It has long since become impractical to wait for the clock to
propagate across a chip, and there are numerous clock domains and
asynchronous interfaces in any off-the-shelf ``synchronously'' clocked
computer~\cite{foster2015trends}. On the other hand, delay-insensitive
circuits~\cite{martin1986compiling} suffer from substantial computational
limitations~\cite{martin1990limitations,manohar2017eventual,manohar2019asynchronous}
and provide no timing guarantees, rendering them
unsuitable for many applications -- in particular the construction of a
general-purpose computer. Accordingly, most real-world ``asynchronous'' systems
will utilize timing assumptions on some components, which in fact could be used
to construct a (possibly very primitive) clock.

As systems grow in size -- physically or due to further miniaturization --
maintaining the illusion of perfect synchronism becomes increasingly
challenging. Due to various scalability issues, more and more compromises are
made. A well-known such compromise gaining in popularity in recent years are
Globally Asynchronous Locally Synchronous (GALS) systems~\cite{chapiro1984globally,teehan2007survey}. Here,
several clock domains are independently clocked and communicate asynchronously
via handshakes, where synchronizers are used to ensure sufficiently reliable
clock domain crossing~\cite{teehan2007survey,dobkin2004data}. While this approach resolves important scalability
issues, arguably it does so by surrendering to them: between clock domains, all
interaction is asynchronous. However, fixing a sufficiently small probability of
synchronizer failure, communication latency becomes bounded, permitting bounded
response times to internal and external events. Yet, as timing relations between
different clock domains remain desirable, GALS systems with guaranteed frequency
relations between clock domains (but without any bound on their phase offsets), so-called mesochronous architectures, have been
conceived~\cite{teehan2007survey}.

One might think that GALS systems exemplify a fundamental struggle between the
synchronous and asynchronous paradigms. We argue that this dichotomy is false!
Rather, choices between clocked and clockless designs are driven by tradeoffs
between guarantees on response times, cost (in terms of energy, buffer size, area, etc.),
and complexity of development. Ideally, we would like to provide the convenient
synchronous abstraction to the developer, yet have the system respond quickly to
external and internal events. Unfortunately, existing approaches behave less
than ideal in this regard:
\begin{itemize}[leftmargin=*]
  \item Centralized clocking does not scale. In large systems, the resulting
  timing guarantees become too loose (requiring to make the system slow).
  Indeed, it has been shown that the achievable \emph{local skew}, i.e.,
  maximum phase offset between neighbors,
  in a grid \emph{grows linearly} with the width of
  the grid; see Section~\ref{subsec:clocktree}.

  \item A system-wide asynchronous design results in challenging
  development, especially when tight timing constraints are to be met. While in
  a clocked system one can bound response times by bounding the number of
  clock cycles for computation and communication, analyzing the
  (worst-case) response time of a large-scale asynchronous system has to be
  performed bottom-up.
  In addition, without highly constraining design rules, it
  is difficult to ensure that waiting for acknowledgements
  does not delay the response to a high-priority local event or an external
  request for a significant time.
  Causal acknowledge chains can span the entire system, potentially resulting
  in \emph{waiting times} that \emph{grow linearly} with the system diameter.

  \item A GALS design ostensibly does not suffer from these issues, as each
  clock domain can progress on its own due to independent clocks.\footnote{This is different for designs with pausible clocks~\cite{yun1996pausible,fan2009analysis},
  rendering them even more problematic in this context.}
  However, clock domain crossings require synchronizers, incurring $2$ or more clock cycles of
  additional latency.
  If synchronizers are placed in the data path,
  communication becomes slow, even if a simple command
  is to be spread across the chip or information is acquired from an
  adjacent clock domain.

  \item Alternative solutions that do not require synchronizers in the data path
    have been proposed in \cite{DDX95low,CG03efficient}.
    The designs either skip clock cycles or switch to a clock
    signal shifted by half a period, when transmitter and receiver clock risk
    to violate setup/hold conditions.
    The indicating signal is synchronized without additional latency to the
    datapath.
    Depending on the implementation and intended guarantees, the additional latency
      is in the order of a clock period.
    While this can, in principle, be brought down to the order of setup/hold-windows,
      such designs would require considerable logical overhead and fine-tuning of
      delays.
    Further, note that an application of such a scheme has to periodically
    insert no-data packets. An application-level transmission may be delayed by such a timeslot.
    In \cite{DDX95low} this additional delay can be up to two periods when the
      no-data packet is oversampled.
    Finally, note that a potential application that runs on top of this scheme
      and uses handshaking to make sure all its packets of a (logical) time
      step have arrived before the next time step is locally initiated faces
      the same problem as a fully asynchronous design, i.e., that the worst-case
      waiting time between consecutive time steps grows \emph{linearly} with the
      system diameter.

\end{itemize}
%###
\paragraph*{Our Contribution}
%###

In this work, we present a radically different approach. By using a distributed
clock synchronization algorithm, we essentially create a single, system-wide
clock domain without needing to spread a clock signal from a single dedicated
source with small skew. We employ results on \emph{gradient clock
synchronization (GCS)} by Lenzen et~al.~\cite{lenzen10tight}, in which the goal
is to minimize the worst-case clock skew between adjacent nodes in a network.
In our setting, the modules correspond to nodes, and they are connected
by an edge if they directly communicate (i.e., exchange data).
Thus, nodes of the clock synchronization algorithm communicate only if the respective
nodes exchange data for computational functionality.
This leads to an easy integration of our algorithm into the existing communication
infrastructure.

The algorithm provides strong parametrized guarantees.
Consider a network of local clocks that are controlled by our GCS algorithm.
Let $D$ be the diameter of the network.
Further, let $\rho$ be the (unintended) drift of the local clock, $\mu > 2 \rho$ a
  freely chosen constant, and $\delta$ an upper bound on how
  precisely the phase difference to neighbors is known.
Then:
\begin{itemize}[leftmargin=*]
  \item The synchronized clocks are guaranteed to run at normalized rates
  between $1$ and $(1+\mu)(1+\rho)$.
  \item The \emph{local skew} is bounded by $O(\delta \log_{\mu/\rho}D)$.
  \item The \emph{global skew}, i.e., the maximum phase offset between any two
  nodes in the system, is $O(\delta D)$.
\end{itemize}
In other words, the synchronized clocks are almost as good as free-running
clocks with drift $\rho$, yet the local skew grows only logarithmically in the chip's diameter.
The local and global skew bounds are optimal up to roughly factor~2 \cite{lenzen10tight}.

As a novel theoretical result, we improve the global skew bound by roughly factor
$2$ compared to~\cite{Kuhn2009-gradient}.
This improvement brings our theoretical worst-case skew to within a factor of roughly
$2$ of the theoretical optimum (which is only known to be achieved by a significantly more 
complicated mechanism~\cite{lenzen10tight}).
As a second theoretical contribution, we prove that a minor
modification of the algorithm reduces the obtained local skew bound by an additive
$2 \delta$.

We can control the base of the logarithm in the local skew bound by choosing
$\mu$. Picking, e.g., $\mu = 100\rho$ means that $\log_{\mu/\rho}D\leq 1$ for any
$D \leq 100$. Of course, the constants hidden in the $O$-notation matter, but
they are reasonably small. Concretely, for a grid network of $32 \times 32$
nodes in the $15$\,nm FinFET-based Nangate OCL~\cite{martins2015open}, $2$\,GHz
clock sources with an assumed drift of $\rho=10^{-5}$, and $\mu =
10^{-3}$, our simple sample implementation guarantees that $\delta\leq
5\,\text{ps}$ in the worst case. The resulting local skew is $30$\,ps, well
below a clock cycle. We stress that this enables much faster communication than
for handshake-based solutions incurring synchronizer delay.

Note that locking the local oscillators to a common stable reference does \emph{not}
require to balance the respective path delays, implying that our assumed $\rho$ is very
pessimistic. Smaller $\rho$ (while keeping $\mu$ fixed) increases the base of the
logarithm, further improving scalability. To show that the asymptotic behavior is
relevant already to current systems and with our pessimistic $\rho$, we compare the
above results to skews obtained by clock trees in the same grid networks in
Section~\ref{subsec:clocktree}.

%###
\paragraph*{Organization of this paper}
We present the GCS algorithm in Section~\ref{sec:algo}, stating worst-case bounds
on the local and global skews proved in the
\ifarxiv
appendix
\else
arXiv version~\cite{arxiv}.
\fi
We then break down the algorithm into modules in Section~\ref{sec:modules} and discuss
their implementation in Section~\ref{sec:implementation}.
Section~\ref{sec:simulation} presents Spice simulations for a network of four nodes,
organized in a line and compares them to clock trees.
We conclude in
Section~\ref{sec:conclusion}.
%###

%% file: algo.tex
\section{Algorithm}\label{sec:algo}
\subsection{High-level Description}

We give a high level description of our algorithm that achieves close synchronization between neighboring nodes in a network. We model the network as an undirected graph $G = (V, E)$ where $V$ is the set of nodes, and $E$ is the set of edges (or links). Abstractly, we think of each node $v$ as maintaining a \dft{logical clock}, which we view as a function $L_v\colon \nnR \to \R$. That is for each (Newtonian) time $t$, $L_v(t)$ is $v$'s logical clock value at time $t$. The \dft{local skew} is the maximum clock difference between neighbors: $\calL(t) = \max_{\set{v, w} \in E} \set{\abs{L_v(t)-L_w(t)}}$. The \dft{global skew} is the maximum clock difference between any two nodes in the network: $\calG(t) = \max_{v, w \in V} \set{\abs{L_v(t) - L_w(t)}}$. The goal of our algorithm is for each node to compute a logical clock $L_v(t)$ minimizing $\calL(t)$ at all times $t$, subject to the condition that all logical clocks progress at least at (normalized) rate $1$.\footnote{Without the minimum rate requirement, the task becomes trivial: all nodes can simply set $L_v(t) = 0$ for all times $t$ to achieve perfect ``synchronization.''}

We assume that each node $v$ has an associated reference clock signal, which we refer to as $v$'s \dft{hardware clock}, denoted $H_v(t)$. For notational convenience,\footnote{It is common to assume a two-sided frequency error, i.e., a rate between $1-\rho$ and $1+\rho$. However, the one-sided notation simplifies expressions. Translating between the two models is a straightforward renormalization.} we assume that the minimum (normalized) rate of $H_v$ is $1$, and its maximum rate is $1 + \rho$: for all $v \in V$ and $t, t' \in \nnR$
\begin{equation}
  \label{eqn:hardware-rate}
  t' - t \leq H_v(t') - H_v(t) \leq (1 + \rho) (t' - t).
\end{equation}
To compute a logical clock, after initially setting $L_v(0) = H_v(0)$, $v$ adjusts the rate of $L_v$ relative to the rate of $H_v$ (where this rate itself is neither known to nor under the influence of the algorithm). Specifically, $v$ can be either in slow mode or fast mode. In slow mode, $L_v$ runs at the same rate as $H_v$, while in fast mode, $v$ sets the rate of $L_v$ to be $1 + \mu$ times the one of its hardware clock. Here, $\mu$ is a parameter fixed by the designer. In order for the algorithm to work, a fast node must always run faster than a slow node---i.e., $\mu > \rho$. We impose the stronger condition that $\mu > 2\rho$.

The GCS algorithm of Lenzen et al.~\cite{lenzen10tight} specifies conditions for a node to be in slow or fast mode that ensure asymptotically optimal local skew, provided that the global skew is bounded. The algorithm is parametrized by a variable $\kappa \in \R^+$, whose value determines the quality of synchronization.

\begin{definition}
  \label{dfn:fast-slow-cond}
  Let $\kappa \in \R^+$ be a parameter. We say that a node $v$ satisfies the \dft{fast condition} at time $t$ if there exists a natural number $s \in \N$ such that the following two conditions hold:
  \begin{description}
  \item[FC1] $v$ has a neighbor $x$ such that $L_x(t) - L_v(t) \geq (2 s + 1) \kappa$
  \item[FC2] all of $v$'s neighbors $y$ satisfy $L_v(t) - L_y(t) \leq (2 s + 1) \kappa$.
  \end{description}
  It satisfies the \dft{slow condition} if there exists $s \in \N$ such that:
  \begin{description}
  \item[SC1] $v$ has a neighbor $x$ such that $L_v(t) - L_x(t) \geq 2 s \kappa$
  \item[SC2] all of $v$'s neighbors $y$ satisfy $L_y(t) - L_v(t) \leq 2 s \kappa$.
  \end{description}
\end{definition}

\begin{definition}
  \label{dfn:gcs-implementation}
  We say that an algorithm \dft{is a GCS algorithm} with parameters $\rho, \mu, \kappa$
  if the following invariants hold, for every node $v \in V$ and all times $t, t'$:
  \begin{description}
  \item[I1] $\mu > \rho$,
  \item[I2] $H_v(t') - H_v(t)\! \le\! L_v(t') - L_v(t)\! \le\! (1 + \mu)(H_v(t') - H_v(t))$
  \item[I3] if $v$ satisfies the \emph{fast condition} throughout the interval $[t, t']$, then $L_v(t') - L_v(t) = (1 + \mu)(H_v(t') - H_v(t))$
  \item[I4] if $v$ satisfies the \emph{slow condition} throughout the interval $[t, t']$, then $L_v(t') - L_v(t) = H_v(t') - H_v(t)$.
  \end{description}
\end{definition}
Invariants (I3) and (I4) still allow a node's clock $L_v(t)$ to vary within the rates of the underlying hardware clock,
  which is assumed not be under the control of the algorithm.

\begin{theorem}\label{thm:gcs}
Suppose algorithm $A$ is a GCS algorithm.
Then $A$ maintains global skew $\calG(t) \leq \frac{\mu \kappa D}{\mu - 2 \rho}$ and
local skew $\calL(t) \leq \left(2 \left\lceil\log_{\mu / \rho} \frac{\mu D}{\mu - 2 \rho}\right\rceil + 1 \right)\kappa$ for all sufficiently large $t$.
\end{theorem}

\begin{remark}
The precise local and global skew bounds achieved by a GCS algorithm at an arbitrary time $t$ depend on the initial state of the system. GCS algorithms are self-stabilizing in the sense that starting from an arbitrary initial state, the algorithm will eventually achieve the skew bounds claimed in Thm.~\ref{thm:gcs} (see~\cite{Kuhn2010-optimal}).
In the
\ifarxiv
appendix,
\else
arXiv version,
\fi
we analyze the speed of convergence as function of local skew at initialization.
\end{remark}

In order to fulfill the invariants of a GCS algorithm, each node $v$ maintains estimates of the offsets to neighboring clocks. Specifically, for each neighboring node $w$, $v$ computes an offset estimate $\offset_w(t) \approx L_w(t) - L_v(t)$. Given offset estimates for each neighbor, the synchronization algorithm determines if $v$ should run in fast mode by checking if the fast trigger ($\FT$) is satisfied, as defined below. The trigger is parametrized by variables $\kappa$ (as in the GCS algorithm) and $\delta$, whose values are determined by the quality of estimates of neighboring clock values.

\begin{definition}
  \label{dfn:trigger}
  We say that $v$ satisfies the \dft{fast trigger}, $\FT$, if there exists $s \in \N$ such that the following conditions hold:
  \begin{description}
  \item[FT1] $\omax \geq (2 s + 1) \kappa - \delta$,
  \item[FT2] $\omin \geq - (2 s + 1) \kappa - \delta$.
  \end{description}
\end{definition}

We are now in the position to formalize our GCS algorithm, $\OGCS$ (Algorithm~\ref{alg:ogcs}). $\OGCS$ is simple: at each time, each node checks if it satisfies $\FT$. If so, it runs in fast mode. Otherwise, the node runs in slow mode. As the decision to run fast or slow is a discrete decision, a hardware implementation will be prone to metastability~\cite{m-gtmo-81}.
We discuss how to work around this problem in Section~\ref{sec:modules}.

\algblockdefx[NAME]{On}{EndOn}%
    [1]{\textbf{At} #1 \textbf{do}}%
    {\textbf{End}}
\algtext*{EndOn}
\begin{algorithm}[ht]
  \begin{algorithmic}[1]
    \On{each time $t$}
      \State $\omin \gets \min_w \{ \offset_w(t) \mid w \text{ is neighbor of } v \}$
      \State $\omax \gets \max_w \{ \offset_w(t) \mid w \text{ is neighbor of } v \}$
      \If{ $v$ satisfies \FT }
        \State \# fast mode (rate in $[(1+\mu), (1+\rho)(1+\mu)]$)
        \State rate of $L_v$ $\gets (1+\mu)\,\, \cdot$ rate of $H_v$
      \Else
        \State \# slow mode (rate in $[1, (1+\rho)]$)
        \State rate of $L_v$ $\gets $ rate of $H_v$
      \EndIf
    \EndOn
  \end{algorithmic}
  \caption{$\OGCS$ algorithm for node $v$}
  \label{alg:ogcs}
\end{algorithm}

In what follows, we show that for a suitable choice of parameters, $\OGCS$ is a GCS algorithm in the sense of Def.~\ref{dfn:gcs-implementation}. Thus, $\OGCS$ maintains the skew bounds of~Thm.~\ref{thm:gcs}.

\subsection{Analysis of the $\OGCS$ algorithm}\label{sec:guar}

We denote an upper bound on the overall uncertainty of $v$'s estimate of $w$ by $\delta$:
\begin{equation}
  \label{eqn:delta-pm}
  \abs{\offset_w(t) - (L_w(t) - L_v(t))} \leq \delta.
\end{equation}

In our analysis, it will be helpful to distinguish two sources of uncertainty faced by any implementation of the GCS algorithm. The first is the \dft{propagation delay uncertainty}, which is the absolute timing variation in signal propagation adding to the measurement error. We use the parameter $\delta_0 > 0$ to denote an upper bound on this value.

The second source of error is the time between initiating a measurement and actually ``using'' it in control of the logical clock speed. During this time, the logical clocks advance at rates that are not precisely known. Here, we can exploit that the maximum rate difference between any two logical clocks is $(1+\rho)(1+\mu)-1=\rho + \mu + \rho \mu$. Thus, denoting the \dft{maximum end-to-end latency} by $\Tmax$, this contributes an error of at most $(\rho + \mu + \rho \mu)\Tmax$ at any given time. Time $\Tmax$ includes the time for the logical clock to respond to control signal.

Once suitable values of $\delta_0$ and $\Tmax$ are determined, $\delta$ can be computed easily.
\begin{lemma}\label{lem:delta}
With $\delta = \delta_0 + (\rho + \mu + \rho \mu) \cdot \Tmax$, Ineq.~\eqref{eqn:delta-pm} holds.
\end{lemma}
Based on $\delta$, we now seek to choose $\kappa$ as small as possible to realize the invariants given in Def.~\ref{dfn:gcs-implementation}. The basic idea is to ensure that if a node $v$ satisfies the fast condition at time $t$ (which depends on the unknown phase difference), then it must satisfy the fast trigger (which is expressed in terms of the estimates $\widehat{O}_w$), thus ensuring that $v$ is in fast mode at time $t$. In turn, if the slow condition is not satisfied, we must make sure that the fast trigger does not hold either.

\begin{lemma}
  \label{lem:uncertainty}
  Suppose for all times $t$ an implementation of $\OGCS$ satisfies (\ref{eqn:delta-pm}). Then for any
\begin{align}
  \kappa &> 2 \delta \label{eqn:kappa}
\end{align}
  and $\mu > \rho$, $\OGCS$ is a GCS algorithm.
\end{lemma}
\begin{proof}
  We verify the conditions of Def.~\ref{dfn:gcs-implementation}. Conditions~\textbf{I1} and~\textbf{I2} are direct consequences of the algorithm specification. For Condition~\textbf{I3}, suppose first that $v$ satisfies the fast condition at time $t$. Therefore, there exists some $s \in \N$ and neighbor $w$ of $v$ such that $L_w(t) - L_v(t) \geq (2 s + 1) \kappa$. Therefore, by Ineq.~\ref{eqn:delta-pm}, $\offset_w(t) \geq (2 s + 1) \kappa - \delta$, so that \textbf{FT1} is satisfied. Similarly, since $v$ satisfies the fast condition, all of its neighbors $x$ satisfy $L_v(t) - L_x(t) \leq (2 s + 1) \kappa-\delta$. Therefore, $\offset_x(t) \geq -(2 s + 1) \kappa$, hence \textbf{FT2} is satisfied for the same value of $s$ and $v$ runs in fast mode at time $t$.

It remains to show that if $v$ satisfies the slow condition at time $t$, then it does not satisfy $\FT$ at time $t$ (and, accordingly, is in slow mode). To this end suppose to the contrary that $v$ satisfies $\FT$ at $t$. Since $v$ satisfies the slow condition at time~$t$,
  \begin{align}
    &\exists x \colon L_v(t) - L_x(t) \geq 2 s \kappa - \delta \label{eqn:sc-1}\\
    &\forall y \colon L_y(t) - L_v(t) \leq 2 s \kappa + \delta. \label{eqn:sc-2}
  \end{align}
  Since $v$ is assumed to satisfy $\FT$ at time $t$, combining FT1 and FT2 with (\ref{eqn:delta-pm}) imply that there exists some $s' \in \N$ with
  \begin{align}
    &\exists x \colon L_x(t) - L_v(t) \geq (2 s' + 1) \kappa - \delta \label{eqn:ft-1}\\
    &\forall y \colon L_v(t) - L_y(t) \leq (2 s' + 1) \kappa + \delta. \label{eqn:ft-2}
  \end{align}
  Combining~(\ref{eqn:sc-2}) and~(\ref{eqn:ft-1}), we must have
  \[
  (2 s' + 1) \kappa - \delta \leq 2 s \kappa + \delta,
  \]
  hence
  $
  2 s' \kappa \leq 2 s \kappa - \kappa + 2\delta
  $.
  Since $2\delta < \kappa$, the previous expression implies that
  $s' < s$.
  Similarly, combining~(\ref{eqn:sc-1}) and~(\ref{eqn:ft-2}) gives
  $
  2 s \kappa - \delta \cdot \Tmax \leq (2 s' + 1) \kappa + \delta
  $,
  hence
  $
    2 s \kappa \leq 2 s' \kappa + 2\delta < 2 (s' + 1) \kappa
  $.
  Thus, $s < s' + 1$, or equivalently (since $s$ and $s'$ are integers), that $s \leq s'$. However, this final expression contradicts $s' < s$ from before. Thus $\FT$ cannot be satisfied at time $t$ if the slow condition is satisfied at time $t$, as desired.
\end{proof}

Applying Thm.~\ref{thm:gcs} and Lem.~\ref{lem:uncertainty} we obtain:

\begin{corollary}\label{cor:local-skew}
  For suitable choices of parameters, $\OGCS$ maintains local skew
  \[
  \calL(t) \leq \left(2 \left\lceil\log_{\mu / \rho} \paren{\frac{\mu \cdot D}{\mu - 2 \rho}}\right\rceil + 1 \right)\kappa.
  \]
\end{corollary}

%% file: modules.tex
\section{Modules}
\label{sec:modules}
For a hardware implementation of the $\OGCS$ algorithm, we break down the distributed algorithm into modules.
Per node, this will be a local clock and a controller.
Per link, we have a time offset measurement module for each node connected via the link.
For each module we specify its input and output ports, its functionality, and its delay.
We further relate the delay $\Tmax$ from Section~\ref{sec:algo} to the module delays.

\subsection{Local Clock}
The clock signal of node $v$ is derived from a tunable local clock oscillator.
It has input $\modes_v$, the mode signal (given by the controller; see Section~\ref{sec:conts}), and output $\text{CLK}_v$, the clock signal.
The mode signal $\modes_v$ is used to tune the frequency of the oscillator within a factor of $1+\mu$.
An oscillator responds within time $\Tosc \geq 0$, i.e., switching between the two
frequency modes takes at most $\Tosc$ time.
We have four requirements to the local clock module:
\begin{description}
  \item [(C1)] The initial maximum local skew is bounded by $c \cdot \kappa$ for a parameter $c > 0$ depending on the implementation of the module.
  \item [(C2)] If $\modes_v$ is constantly $0$ (respectively $1$) during $[t-\Tosc,t]$, then the local oscillator is in \emph{slow} (respectively \emph{fast}) \emph{mode} at time $t$ and the rate of the local oscillator is in $[1,1+\rho]$ (respectively $[1+\mu,(1+\mu)(1+\rho)]$).
  \item [(C3)] If $\modes_v$ is neither constantly $0$ nor $1$ during $[t-\Tosc,t]$, then the local oscillator is unlocked and its rate is in $[1,(1+\mu)(1+\rho)]$.
  \item [(C4)] Clocks in slow mode are never faster than clocks in fast mode, hence $\mu > \rho$.
\end{description}
Note that if (C2) does not apply, i.e., the mode signal is not stable, (C3) allows an arbitrary rate between fast and slow.

\subsection{Time Offset Measurement}
In order to check whether the $\FT$ conditions are met, a node $v$ needs to measure the current phase offset $\offset_w$ to each of its neighbors $w$.
This is achieved by a time offset measurement module between $v$ and each neighbor $w$.
Note that the algorithm does not require a full access to the function $\offset_w$, but only to the knowledge of whether  $\offset_w$ has reached a bounded number of thresholds -- we elaborate on this shortly.

The inputs of the module are the clock signal of $v$ and $w$.

The outputs of the module are defined as follows.
Let $S=\{0, \ldots,\ell\}$ with $\ell > 0$.
The output of the measurement module is a binary string of length $2(\ell+1)$ bits where the first $\ell+1$ bits, denoted as $Q^{i}_w$, are going from $\ell$ to $0$, followed by additional $\ell+1$ bits, denoted as $Q^{-i}_w$, going from $0$ to $\ell$.
For example, a module with $S = \{0,1\}$ has $4$ outputs
  with thresholds $3\kappa-\delta$, $\kappa-\delta$, $-\kappa-\delta$, and $-3\kappa-\delta$.

Let $\e > 0$ be a (small) time.
We require that output $Q^{\pm i}_w$ is set to $1$ if $\offset_w(t) \geq \mp(2 (i-1) + 1) \kappa - \delta + \e$.
Output $Q^{\pm i}_w$ is set to $0$ if $\offset_w(t) \leq \mp(2 (i-1) + 1) \kappa - \delta$.
Otherwise, $Q^{\pm i}_w$ is unconstrained, i.e., within $\{0,\metas,1\}$.
Here, $\metas$ denotes a meta-/unstable signal between logical values~$0$ and~$1$.
Intuitively, $\e$ will account for setup/hold times that any realistic hardware implementation will have to account for.

We further require that $\e < 2\kappa$. This guarantees that at most one output is $\metas$ at a time:
Assume that bit $Q^{i}_w$ is metastable,
then $\offset_w(t) \in (2 (i-1) + 1) \kappa - \delta + [0,\e]$.
Since the adjacent thresholds are $2\kappa$ away, their corresponding outputs are either $0$ or $1$. In fact, by Eq.~\eqref{eqn:kappa} and since $\e \leq \delta_0$ (we account for setup/hold times in $\delta_0$), we get that $\e < \frac{\kappa}{2}$, hence our requirement is satisfied.

Choosing $\ell \ge \frac{\calL/\kappa - 1}{2}$, where $\calL$ is the guaranteed local skew of the $\OGCS$ algorithm, guarantees that the nodes will always be within the module's measurement range. Note that $\calL$ here needs to respect the initial local skew as well, i.e., $\calL$ here is given by the bound from Cor.~\ref{cor:local-skew} plus the local skew on initialization
\ifarxiv
(as we show in the appendix)
\fi
.

Given the above, the module outputs form a unary thermometer code of the phase difference between $v$ and $w$'s clocks.
Moreover, since this module decides whether a subset of the thresholds are met or not, then inevitably, any implementation of this module (see Section~\ref{sec:imp}) is susceptible to metastable upsets. If implemented correctly, one can leverage the output encoding, which is a unary thermometer code, and guarantee that at most one bit is in a metastable state, located conveniently between a prefix of $1$'s and a suffix of $0$'s.

Let $\Tmeas$ denote the maximum end-to-end latency of the measurement module, i.e.,
an upper bound on the elapsed time from when $Q^{\pm i}_w$ is set, to when the measurements are available at the output.
More precisely, if $Q^{\pm i}_w$ is set to $x \in \{0,1\}$ for the entire duration of an interval $[t - \Tmeas, t]$,
  then the corresponding output is $x$.

\subsection{Controller}\label{sec:conts}
Each node $v$ is equipped with a controller module.
Its input is the (thermometer encoded) time measurement for each of $v$'s neighbors, i.e., the outputs of the
time offsets measurement module on each link connecting $v$ to an adjacent node.
It outputs the mode signal $\modes_v$.

Denote by $\Tctr$ the maximum end-to-end delay of the controller circuit, i.e., the delay between its inputs (the measurement offset outputs) and its output $\modes_v$. The specification of the controller’s interface is as follows:
\begin{description}
 \item [(L1)] For $t > \Tctr$, if algorithm $\OGCS$ continuously maps the rate of
 $v$ to fast mode (resp.\ slow mode) during $[t-\Tctr, t]$, then $\modes_v(t) = 1$ (resp.\ $\modes_v(t)=0$)
 \item [(L2)] In all other cases, the output at time $t$ is arbitrary, i.e., any value from $\{0,\metas,1\}$.
\end{description}

\subsection{Putting it all together}
The module specifications above, together, specify a realization of the $\OGCS$ algorithm in hardware.
The parameters of this hardware specification of $\OGCS$ are: $\delta_0, \rho, \mu$, and $\Tmax$, where $\Tmax = \Tmeas + \Tctr + \Tosc$. These parameters are mapped to parameters of Cor.~\ref{cor:local-skew} by applying Lem.~\ref{lem:uncertainty}.

%% file: implementation.tex
\section{Hardware Implementation}\label{sec:imp}
\label{sec:implementation}

We have implemented the modules from Section~\ref{sec:modules} and compiled them into a system of $4$ nodes, connected in a line from node $0$ to node $3$.  To resemble a realistically sparse spacing of clocks, we placed nodes at distances of $200$\,$\mu$m.  Target technology was the $15$\,nm FinFET-based Nangate OCL~\cite{martins2015open}.  The gate-level design was laid out and routed with Cadence Encounter, which was also used for extraction of parasitics and timing.  Local clocks run at a frequency of approximately $2$\,GHz, controllable within a factor of $1 + \mu \approx 1 + 10^{-4}$. We use $\mu=10^{-4}$ here to make the interplay of $\rho$ and $\mu$ better visible in traces. We will discuss the gate-level design and its performance measures in the following.

\subsection{Gate-level Implementation}
Figures~\ref{fig:measdline} to~\ref{fig:mode} show the schematics of an implementation of the time offsets measurement module
(Figure~\ref{fig:measdline}), and the controller (Figures~\ref{fig:minmax} and \ref{fig:mode}).

\begin{figure*}[ht]
    \centering
    \begin{subfigure}[b]{0.42\textwidth}
        \centering
        \includegraphics[width=\textwidth]{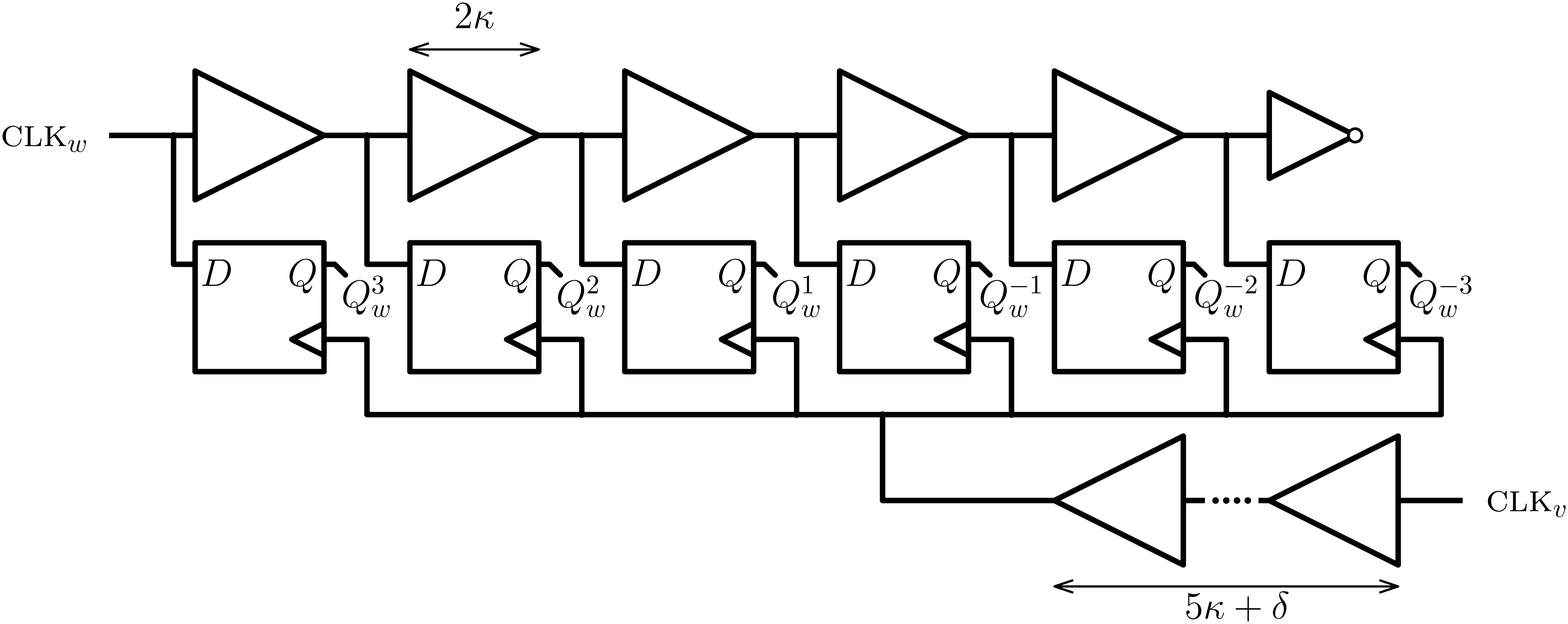}
        \caption{}
        \label{fig:measdline}
    \end{subfigure}
    \quad%\hfill
    \begin{subfigure}[b]{0.15\textwidth}
        \centering
        \includegraphics[width=\textwidth]{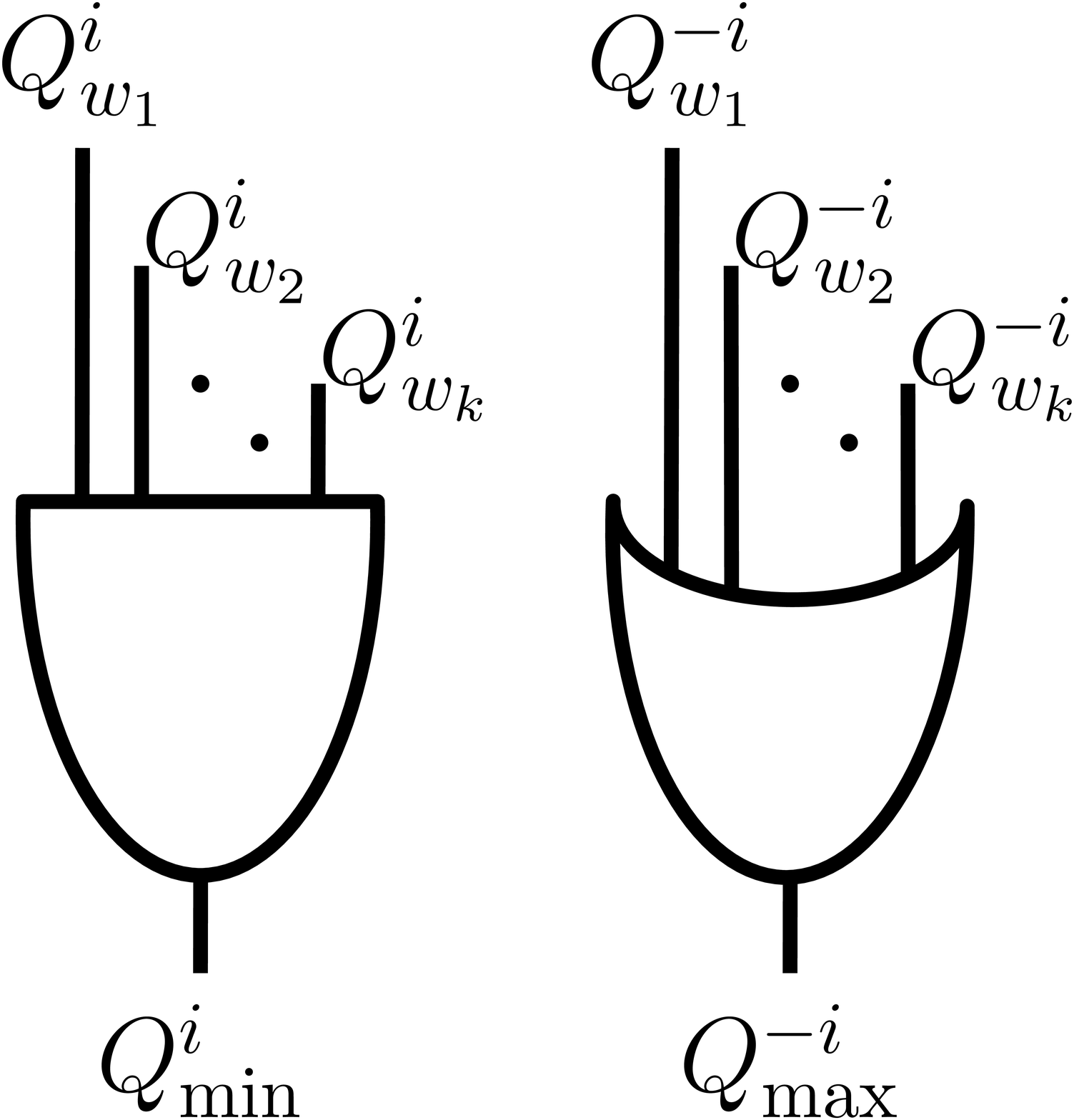}
        \caption{}
        \label{fig:minmax}
    \end{subfigure}
    \quad%\hfill
    \begin{subfigure}[b]{0.30\textwidth}
        \centering
        \includegraphics[width=\textwidth]{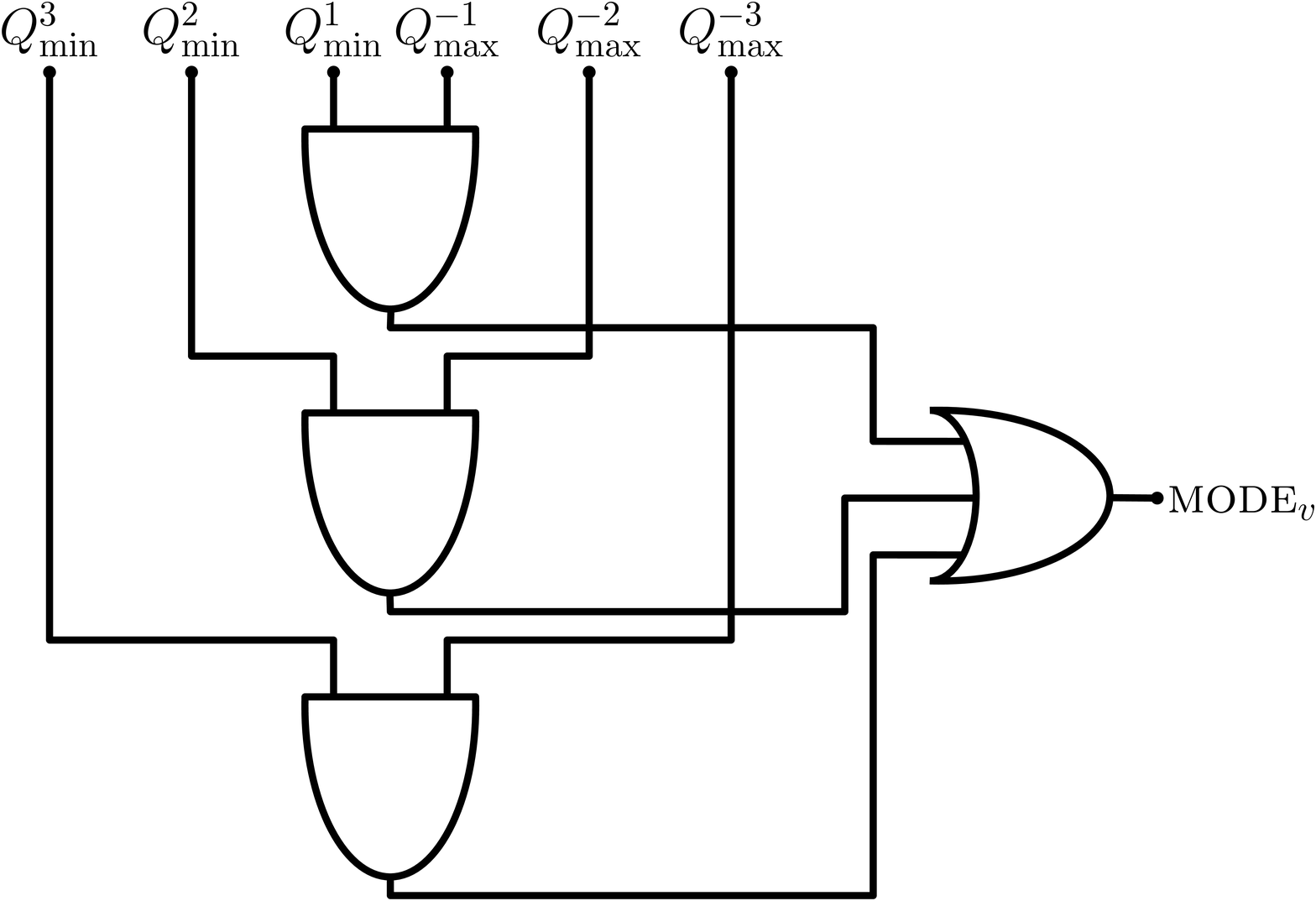}
        \caption{}
        \label{fig:mode}
    \end{subfigure}
    \caption{Gate-level implementation of the $\OGCS$ algorithm's modules.
    Sub-figure~\ref{fig:measdline} shows a linear TDC-based circuitry for the
      module which measure the time offsets between nodes $v$ and $w$.
    Buffers and inverters are used as delay elements the delay of which appears next to the corresponding delay element.
    Given node $v$'s time offsets to its neighbors, the circuit in
      Sub-figure~\ref{fig:minmax} computes the minimum and maximum threshold levels which have been reached.
    Sub-figure~\ref{fig:mode} shows the circuit that computes if the $\FT$ conditions are satisfied, i.e.,
      if there is an $S \in \{0,1,2\}$ that satisfies both {\bf FT1} and {\bf FT2}.
      \vspace{-0.3cm}}
\end{figure*}

As a local clock source, we used a ring oscillator with some of its inverters being starved-inverters to set the frequency to either fast mode or slow mode. Nominal frequency is around $2$\,GHz, controllable by a factor $1 + \mu \approx 1 + 10^{-4}$ via the $\modes_v$ signal. We choose $\rho \approx \mu/10\approx 10^{-5}$, assuming a moderately stable oscillator.  While this is below drifts achievable with uncontrolled ring oscillators, one may lock the \emph{frequency} of the ring oscillator to a stable external quartz oscillator, see e.g.,~\cite{mota2000flexible}. For such an implementation, we only require a stable frequency reference for local clocks; the phase difference of the distributed clock signal between adjacent nodes (which may be large) is immaterial.  If distributing a stable clock source to all nodes is not feasible or considered too costly for a design, one may choose a larger $\mu$ resulting in a larger local and global skew bound; see Thm.~\ref{thm:gcs} and Cor.~\ref{cor:local-skew}.

We measure the logical clock value $L_v(t)$ in terms of the time passed since its first active clock transition.

The time offset measurement module resembles a time to digital converter (TDC) in both its structure and function.  The upper delay line in Figure~\ref{fig:measdline}, fed by remote clock $w$, is tapped at intervals of $2\kappa$.  The lower delay line is used to shift the module's own local clock $v$ to the middle of the delay line (plus some $\delta$ offset) so that phase differences can be measured both in the negative and positive direction.  The module in Figure~\ref{fig:measdline} is instantiated for $S = \{0,1,2\}$ with $6$ taps for threshold levels. In fact, in our hardware implementation we set $S=\{0,1\}$, as even for $\mu/\rho=10$ this is sufficient for networks of diameter up to around $80$ (see how to choose this set of thresholds in the specification of this module in Sec.~\ref{sec:modules}).

If both clocks are perfectly synchronized, i.e., $L_v = L_w$, then the state of the flip-flops will
  be $Q^{3}Q^{2}Q^{1}Q^{-1}Q^{-2}Q^{-3} = 111000$ after a rising transition of $\text{CLK}_v$.
Now, assume that clock $w$ is earlier than clock $v$, say by a small $\e > 0$ more than $\kappa + \delta$\,ps.  Then $L_w = L_v + \kappa - \delta + \e$.  For the moment assuming that we do not make a measurement error, we get $\offset_w = L_w - L_v = \kappa - \delta + \e$.  From the delays in Figure~\ref{fig:measdline} one verifies that in this case, the flip-flops are clocked before clock $w$ has reached the second flip-flop with output $Q^1$, resulting in a snapshot of $110000$.  Likewise, an offset of $\offset_w = L_w - L_v = 3\kappa - \delta + \e$ results in a snapshot of $100000$, etc.

However, care has to be taken for non-binary outputs.
Given the output specification above, one can verify that measurements are of the
  form $1^*0^*$ or $1^*\metas0^*$.

The circuit in Figure~\ref{fig:minmax} then computes the minimum and the maximum of the thermometer codes
  (by AND and OR gates), determining the thresholds reached by the furthest node ahead and behind $v$ (while possibly masking metastable bits);
  compare this with lines~2 and~3 in $\OGCS$ (Algorithm~\ref{alg:ogcs}).
Figure~\ref{fig:mode} finally computes the mode signal of $v$ from the thermometer codes, namely verifying whether there is an $s\in\{0,1,2\}$ that satisfies \emph{both} triggers; compare this with {\bf FT1} and {\bf FT2} in Def.~\ref{dfn:trigger}.

%###
\paragraph{Timing Parameters}
%###
We next discuss how the modules' timing parameters relate to the extracted physical timing of the above design.

The time required for switching between oscillator modes $\Tosc$ is about the delay of the ring oscillator, which in our case is
  about $1/(2\cdot 2\,\text{GHz}) = 250$\,ps.
The measurement latency $\Tmeas$ plus the controller latency $\Tctr$ are given by
  a clock cycle ($500$\,ps) plus the delay ($25$\,ps) from the flip-flops through the AND/OR circuitry in Figures~\ref{fig:minmax}
  and~\ref{fig:mode} to the mode signal.
In our case, delay extraction of the circuit yields $\Tmeas+\Tctr < 500\,\text{ps} + 25\,\text{ps}$.
We thus have, $\Tmax < \Tmeas + \Tctr + \Tosc = 775$\,ps.

The propagation delay uncertainty, $\delta_0$, in measuring if $\offset_w$ has reached a certain threshold is given by the
  uncertainties in latency of the upper delay chain plus the lower delay chain in Figure~\ref{fig:minmax}.
For the described naive implementation using an uncalibrated delay line, this would be problematic.
With an uncertainty of $\pm 5\%$ for gate delays, and starting with moderately sized $\kappa$ and thus length of delay
  chains, extraction of minimum and maximum delays showed that the constraints for $\delta$ and $\kappa$ from
  Lem.~\ref{lem:uncertainty} were not met.
Successive cycles of increasing $\delta$ and $\kappa$ do not converge due to the linear dependency of
  $\delta$ and $\kappa$ on the uncertainty $\delta_0$ with a too large factor.
Rather, delay variations (of the entire system) have to be less than $\pm 1\%$ for the linear offset measurement circuit, depicted in Sub-figure~\ref{fig:measdline},
  to fulfill Lem.~\ref{lem:uncertainty}'s requirements.

\subsection{Improvements}
Figure~\ref{fig:bettermeas} shows an improved TDC-type offset-measurement circuit that does not suffer from
  the problem above.
Conceptually the TDC of node $v$ that measures offsets  w.r.t. node $w$ is integrated into the local ring oscillator of neighboring node $w$.
If $w$ has several neighbors, e.g., up to $4$ in a grid,
  they share the taps, but have their own flip-flops within node $w$.
The Figure shows a design for $S = \{0,1\}$ with $4$ taps, as used in our setup.

\begin{figure}
\centering
  \includegraphics[width=0.8\columnwidth]{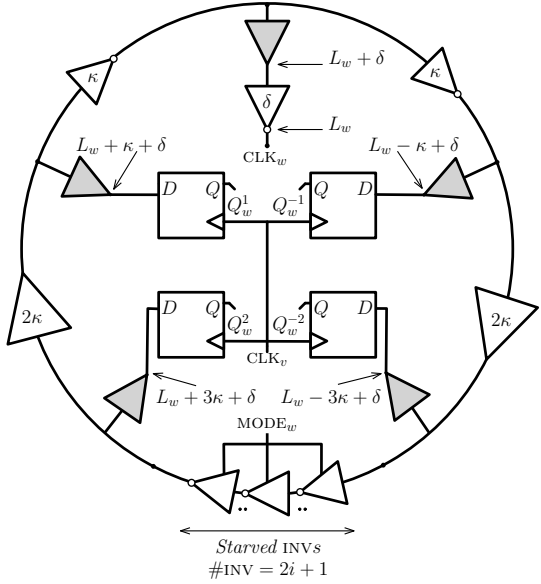}
  \caption{Improved offset measurement implementation. The delays of each delay element are written inside it.
  The gray buffers at the offset measurement taps decouple the load of the remaining circuitry.
  At the bottom of the ring oscillator an odd number of starved inverters used to set slow or fast mode for node $w$.
  The phase offset that we measure in each tap is written next to the corresponding flip-flop.
  The delay elements at the top are inverters instead of buffers to achieve a latency of $\kappa = 10$\,ps.
  We inverted the clock output to account for the negated signal at the tap of clock $w$ at the top.\vspace{-0.3cm}}
\label{fig:bettermeas}
\end{figure}

Integration of the TDC into $w$'s local ring oscillator greatly reduces uncertainties
  at both ends: (i) the uncertainty at the remote clock port (of node $w$) is removed to a large extent, since the
  delay elements which are used for the offset measurements are part of $w$'s oscillator,
  and (ii) the uncertainty at the local clock port is greatly reduced by removing the delay line
  of length $5\kappa + \delta$.
Remaining timing uncertainties are the latency from taps to the D-ports of the flip-flops
  and from clock $v$ to the CLK-ports of the flip-flop.
Timing extraction yielded $\delta_0 < 4$\,ps
  in presence of $\pm 5\%$ gate delay variations.

From Lem.~\ref{lem:uncertainty}, we thus readily obtain
  $\kappa \approx 10$\,ps and $\delta \approx 5$\,ps
  which matched the previously chosen latencies of the delay
  elements.
Applying Thm.~\ref{thm:gcs} and Cor.~\ref{cor:local-skew}
  finally yields a bounds of
  $1.223 \kappa D = 12.23 D$\,ps
  on the global skew and of
  $ (2\lceil \log_{10} (1.223 D) \rceil + 1)\kappa$
  on the local skew.
For our design with diameter $D=3$ this makes
  a maximum global skew of $36.69$\,ps
  and a maximum local skew of $3\kappa = 30$\,ps.
Note that considerably larger systems, e.g., a grid with side length of $W = 32$ nodes
  and diameter $D=2W-2 = 62$, still are guaranteed to have a maximum local
  skew of $3\kappa = 30$\,ps -- and for $\mu=10^{-3}$, the base of the logarithm becomes $100$.

%% file: simulation.tex
\section{Simulation and Comparison to Clock Trees}\label{sec:simulation}

\subsection{Spice Simulations on a Line Topology}

We ran Spice simulations with Cadence Spectre of the post-layout extracted design
for $4$ nodes arranged in a line, as described in Section~\ref{sec:implementation}. The line's nodes
are labeled $0$ to $3$.
For the simulations, we set $\mu = 10\rho$ instead of $100\rho$, resulting in slower
   decrease of skew, to better observe how skew is removed.
We simulated two scenarios where node $1$ is initialized with an offset of
$40$\,ps ahead of (resp.\ behind) all other nodes.
Simulation time is $1000$\,ns ($\approx 2000$ clock cycles) for the first and $600$\,ns
  for the second scenario.

Figure~\ref{fig:sim_clocks} shows the clock signals of nodes $0$ to $3$ at three
points in time for the first scenario: (i) shortly after the initialization,
(ii) around $100$\,ns, and (iii) after $175$\,ns.

For the mode signals, in the first scenario, we observe the following:
Since node $1$ is ahead of nodes $0$ and $2$, node $1$'s mode signal is correctly set to
  $0$ (slow mode) while node $0$ and $2$'s mode signals are set to $1$ (fast mode).
Node $3$ is unaware that node $1$ is ahead since it only observes node $2$.
By default its mode signal is set to slow mode.
When the gap to $2$ is large enough it switches to fast mode.
This configuration remains until nodes $0$ and $2$ catch up to $1$, where they
switch to slow mode, to not overtake node $1$.
Again node $3$ sees only node $2$ which is still ahead and switches only after
it catches up to $2$.

Figure~\ref{fig:sim_skews} (red lines) depicts the dynamics of the maximum local and global
skews for the first scenario.
Observe that, from the beginning the local skew decreases until it
  reaches less than $9$\,ps.
It then remains in an stable oscillatory state where it increases until the algorithm
  detects and reduces the local skew.
This is well below our worst-case bound of $30$\,ps on the local skew.
The global skew first increases, as node $3$ does not switch to fast mode immediately.
Scenario two shows a similar behaviour (blue lines in Figure~\ref{fig:sim_skews}).

\begin{figure}%[h]
\centering
\includegraphics[width=\columnwidth]{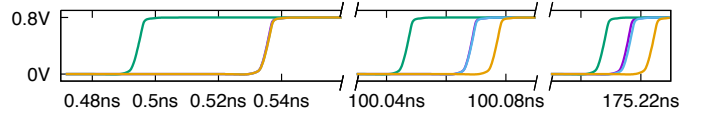} % ,trim=8 3 66 165, clip
\caption{Spice simulation of the line topology. Node $1$ has been initialized with
a skew of $40$\,ps ahead of the other nodes.
Nodes from left to right: (i) $1$ before $0,2,3$, (ii) $1$ before $0,2$ before $3$,
  (iii) $1$ before $0,2$ before $3$.\vspace{-0.2cm}}
\label{fig:sim_clocks}
\end{figure}

\begin{figure}
\centering
\includegraphics[width=\columnwidth]{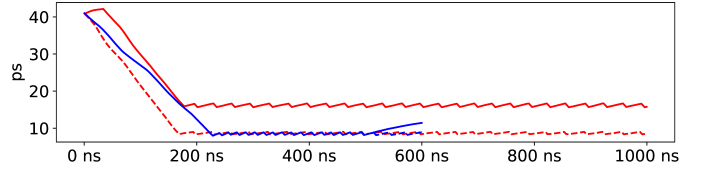}
\caption{Maximum local skew (dotted) and global skew (solid)
for the scenarios of node $1$ initially being ahead (red) and behind (blue)
of all other nodes.\vspace{-0.2cm}}
\label{fig:sim_skews}
\end{figure}

\subsection{Comparison to Clock Tree}
\label{subsec:clocktree}
For comparison, we laid out a grid of $W \times W$ flip-flops,
  evenly spread in $200\,\mu m$ distance in x and y direction across the
  chip.
The data port of a flip-flop is driven by the OR of the up to four
  adjacent flip-flops.
Clock trees were synthesized and routed with Encounter Cadence, with the target
  to minimize local skews.
Delay variations on gates and nets were set to $\pm 5\%$.
The results are presented in Figures~\ref{fig:clk1}.
For comparison, we plotted local skews guaranteed by our algorithm
  for the same grids with parameters extracted from the
  implementation described in Section~\ref{sec:implementation}.
Observe the linear growth of the local clock skew and the
  logarithmic growth of the local skew in our implementation.
  The figure also shows the skew for a clock tree with delay variations of $\pm 10\%$.
  This comparison is relevant, as $\delta_0$ is governed by \emph{local} delay variations,
  which can be expected to be smaller than those across a large chip.

\begin{figure}
\vspace{-0.3cm}
\centering
\includegraphics[width=\columnwidth]{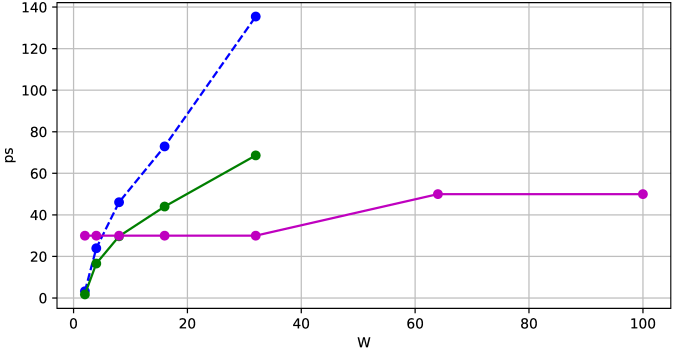}
\caption{Local skew (ps) between neighboring flip-flops in the $W \times W$ grid.
Clock tree with $\pm 5\%$ delay variation (solid green) and our algorithm
with $\pm 5\%$ delay variation (solid magenta).
The dotted line shows the clock tree with $\pm 10\%$ delay variation,
  demonstrating linear growth of the skew also in a different setting.
Clock trees are shown up to $W=32$ after which Encounter ran out of memory.\vspace{-0.2cm}}
\label{fig:clk1}
\end{figure}

It is worth mentioning that it has been shown that no clock
tree can avoid the local skew being
proportional to $W$~\cite{fisher85synchronizing}.

\ifarxiv % ----- arxiv begin
It is worth mentioning that one can show that for \emph{any} clock tree there
are \emph{always} two nodes in the grid that have local skew which is
proportional to $W$. This follows from the fact that there are always two
neighboring nodes in the grid which are in distance proportional to $W$ from
each other in the clock
tree~\cite{fisher85synchronizing,boksberger2003approximation}. Accordingly,
uncertainties accumulate in a worst case fashion to create a local skew which is
proportional to $W$; this behavior can be observed in Figure~\ref{fig:clktree}.

To gain intuition on this result, note that there is always an edge that, if
removed (see the edge which is marked by an X in Figure~\ref{fig:clktree}),
partitions the tree into two subtrees each spanning an area of $\Omega(W^2)$ and
hence having a shared perimeter of length $\Omega(W)$. Thus, there must be two
adjacent nodes, one on each side of the perimeter, at distance $\Omega(W)$ in
the tree.

Our algorithm, on the other hand, manages to reduce the local skew exponentially
to being proportional to $\log W$.

\begin{figure}[h]
\centering
\includegraphics[width=0.9\columnwidth]{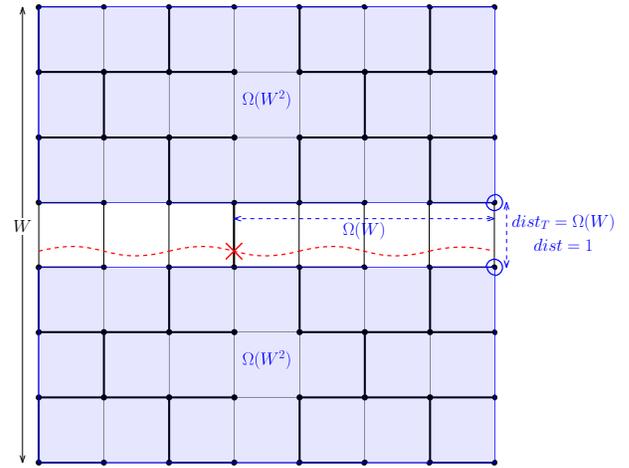}
\caption{A low stretch spanning tree of an $8\times8$ grid~\cite{blog}. The bold lines depict the spanning tree, i.e., our clock tree in this example. The two neighboring nodes that are of distance $13$ in the tree are circled (at the middle right side of the grid).\vspace{-0.3cm}}
\label{fig:clktree}
\end{figure}
\fi % ----- arxiv end

%% file: conclusion.tex
\section{Conclusion}\label{sec:conclusion}

Low skew between neighboring nodes in a chip allows for efficient low-latency communication
  and provides the illusion of a single clock domain.
A classical solution for this problem is to use a clock tree.
However, clock trees inevitably produce local skews which are proportional to
  the diameter of the chip.
We propose a solution based on a distributed clock synchronization algorithm.
Its main idea is to control the local clocks of each node by measuring the time
  offsets from its neighbors and switching between fast and slow clock rates.

We compare our implementation to tool-generated $2$\,GHz clock trees
  for $W \times W$ grids in $15$\,nm technology.
Asymptotically, the implementation improves over the clock tree exponentially.
Our simulations show an improvement of roughly $50\%$ on the local skew
  already for $W=32$.
  
The algorithmic approach is highly robust. It does not rely on a single node or link, and can stabilize to small skews even under poor initialization conditions. In particular, it will recover from transient faults, and can handle the loss of individual nodes or links by adding simple detection mechanisms~\cite{Kuhn2010-optimal}. Moreover, it is known how to integrate new or recovering links or nodes by a simple mechanism without interfering with the skew bounds~\cite{Kuhn2010-optimal}.
Thus, our approach provides a flexible and resilient alternative to classic designs.

In future work, we intend to design a full implementation including suitable (locked) oscillators.
As demonstrated by the work of Mota et al.~\cite{mota2000flexible}, systems with much smaller values of $\rho$ than $10^{-5}$ are feasible.
Consequently, even a simple design is likely to result in sufficiently stable local time references.
However, a challenge here is that the oscillators need to be locked to a (frequency) reference.
This prevents directly adjusting their phase, which would be in conflict with their locking.
This issue can be resolved by using a digitally controlled oscillator derived from the local clock.
Such a design is possible using synchronizers (which however would increase $\Tmax$), or could make use of metastability-containing techniques in the vein of F\"ugger et al.~\cite{fklw-fadcfa-18}.

%% file: proofs.tex
\section{Proof of Lemma~\ref{lem:delta}}
\begin{proof}
Consider the estimate $\widehat{O}_w(t)$ that the algorithm uses at node $v$ for neighbor $w$ at time $t$. By definition of $\Tmax$, the measurement is based on clock values $L_v(t_v)$ and $L_w(t_w)$ for some $t_v,t_w\in [t-\Tmax,t)$. Without loss of generality, we assume that to measure whether $L_w-L_v\ge T\in \R$, the signals are sent at logical times satisfying $L_w(t_w)-T=L_v(t_v)$.\footnote{One can account for asymmetric propagation times by shifting $L_w(t_w)$ and $L_v(t_v)$ accordingly, so long as this is accounted for in $\Tmax$ and carry out the proof analogously.} Denote by $t_v'\in (t_v,t)$ and $t_w'\in (t_w,t)$ the times when the respective signals arrive at the data or clock input, respectively, of the register\footnote{We assume a register here, but the same argument applies to any state-holding component serving this purpose in the measurement circuit.} indicating whether $\widehat{O}_w\ge T$ for a given threshold $T$. By definition of $\delta_0$, we have that
\begin{equation*}
\abs{t_v'-t_v-(t_w'-t_w)}\le \delta_0.
\end{equation*}
Note that the register indicates $\widehat{O}_w(t)\ge T$, i.e., latches $1$, if and only if $t_w'<t_v'$.\footnote{For simplicity of the presentation we neglect the setup/hold time $\e$ (accounted for in $\delta_0$) and metastability;
see Section~\ref{sec:modules} for a discussion.} Thus, we need to show
\begin{align*}
L_w(t)-L_v(t)&\ge T+\delta \implies t_w'<t_v'\\
L_w(t)-L_v(t)&\le T-\delta \implies t_w'>t_v'.
\end{align*}
Assume first that $L_w(t)-L_v(t)\ge T+\delta$. Then, using \textbf{I4} and that $L_w(t_w)-T=L_v(t_v)$, we can bound
\begin{align*}
T+\delta &\le L_w(t)-L_v(t)\\
&\le L_w(t_v)-L_v(t_v)+((1+\mu)(1+\rho)-1)(t-t_v)\\
&= L_w(t_v)-L_w(t_w)+T+(\mu+\rho+\rho \mu)(t-t_v)\\
&\le t_v-t_w+T+(\mu+\rho+\rho\mu)(t-\min\{t_v,t_w\})\\
&<t_v-t_w+T+(\mu+\rho+\rho\mu)\Tmax.
\end{align*}
Hence,
\begin{equation*}
t_w'-t_v'\ge t_w-t_v-\delta_0 > \delta-\delta_0-(\mu+\rho+\rho\mu)\Tmax =0.
\end{equation*}
For the second implication, observe that it is equivalent to
\begin{equation*}
L_v(t)-L_w(t)\ge -T+\delta \implies t_v'>t_w'.
\end{equation*}
As we have shown the first implication for any $T\in \R$, the second follows analogously by exchanging the roles of $v$ and $w$.
\end{proof}

\section{Proof of Theorem~\ref{thm:gcs}}

In this appendix, we prove Theorem~\ref{thm:gcs}.
We assume that at (Newtonian) time $t = 0$, the system satisfies \emph{some} bound on local skew. The analysis we provide shows that the GCS algorithm maintains a (slightly larger) bound on local skew for all $t \geq 0$. An upper bound on the local skew also bounds the number of values of $s$ for which \textbf{FC} or \textbf{SC} (Definition~\ref{dfn:fast-slow-cond}) can hold, as a large $s$ implies a large local skew. (For example, if a node $v$ satisfies \textbf{FC1} for some $s$, then $v$ has a neighbor $x$ satisfying $L_x(t) - L_v(t) \geq (2 s + 1) \kappa$, implying that $\calL(t) \geq (2 s + 1) \kappa$.) Accordingly, an implementation need only test for values of $s$ satisfying $\abs{s} < \frac 1 {2 \kappa} \calL_{\mathrm{max}}$, where $\calL_{\mathrm{max}}$ is an upper bound on the local skew. Our analysis also shows that given an arbitrary initial global skew $\calG(0)$, the system will converge to the skew bounds claimed in Theorem~\ref{thm:gcs} within time $O(\calG(0)/\mu)$. We note that the skew upper bounds of Theorem~\ref{thm:gcs} match the lower bounds of~\cite{lenzen10tight} up to a factor of approximately 2, and these lower bounds apply even under the assumption of initially perfect synchronization (i.e., systems with $\calL(0) = \calG(0) = 0$).

Our analysis also assumes that logical clocks are differentiable functions. This assumption is without loss of generality: By the Stone-Weierstrass Theorem (cf.\ Theorem~7.26 in~\cite{Rudin1976-principles}) every continuous function on a compact interval can be approximated arbitrarily closely by a differentiable function.

We will rely on the following technical result. We provide a proof in Section~\ref{sec:max-bound}.

\begin{lemma}
  \label{lem:max-bound}
  For $k \in \Zpos$ and $t_0, t_1 \in \nnR$ with $t_0 < t_1$, let $\mathcal{F} = \{f_i\,|\,i\in [k]\}$, where each $f_i \colon [t_0, t_1] \to \R$ is a differentiable function. Define $F \colon [t_0, t_1] \to \R$ by $F(t) =  \max_{i\in [k]} \set{f_i(t)}$. Suppose $\mathcal{F}$ has the property that for every $i$ and $t$, if $f_i(t) = F(t)$, then $\frac{d}{dt} f_i(t) \leq r$. Then for all $t \in [t_0, t_1]$, we have $F(t) \leq F(t_0) + r (t - t_0)$.
\end{lemma}

Throughout this section, we assume that each node runs an algorithm satisfying the invariants stated in \cref{dfn:gcs-implementation}. By \cref{lem:delta,lem:uncertainty}, \cref{alg:ogcs} meets this requirement if $\kappa > 2\delta+2(\rho+\mu+\rho\mu)\Tmax$.

\subsection{Leading Nodes}

We start by showing that skew cannot build up too quickly. This is captured by analyzing the following functions.

\begin{definition}[$\Psi$ and Leading Nodes]
  \label{def:psi}
  For each $v\in V$, $s \in \N$, and $t \in \nnR$, we define
  \begin{equation*}
    \Psi_v^s(t) = \max_{w\in V} \{L_w(t) - L_v(t) -  2 s \kappa d(v,w)\},
  \end{equation*}
where $d(v,w)$ denotes the distance between $v$ and $w$ in $G$. Moreover, set
\begin{equation*}
\Psi^s(t) = \max_{v\in V} \{\Psi_v^s(t)\}.
\end{equation*}
Finally, we say that $w\in V$ is a \emph{leading node} if there is some $v\in V$ satisfying
\begin{equation*}
\Psi_v^s(t) = L_w(t) - L_v(t) - 2 s \kappa d(v,w) > 0.
\end{equation*}
\end{definition}

Observe that any bound on $\Psi^s$ implies a corresponding bound on $\calL$: If $\Psi^s(t) \leq \kappa$, then for any adjacent nodes $v, w$ we have $L_w(t) - L_v(t) - 2 s \kappa \leq \Psi^s(t) \leq \kappa$. Therefore, $\Psi^s(t) \leq \kappa \implies \calL \leq (2 s + 1) \kappa$. Our analysis will show that in general, $\Psi^s(t) \leq \Gmax / \sigma^s$ for every $s \in \N$ and all times $t$. In particular, considering $s=\lceil\log_{\mu / \rho} \Gmax/\kappa\rceil$ gives a bound on $\calL$ in terms of $\Gmax$. Because $\calG(t)=\Psi^0(t)$, the skew bounds will then follow if we can suitably bound $\Psi^0$ at all times.

Note that the definition of $\Psi_v^s$ is closely related to the definition of the slow condition. In fact, the following lemma shows that if $w$ is a leading node, then $w$ satisfies the slow condition. Thus, $\Psi^s$ cannot increase quickly: \textbf{I4} (Def.~\ref{dfn:gcs-implementation}) then stipulates that leading nodes increase their logical clocks at rate at most $1+\rho$. This behavior allows nodes in fast mode to catch up to leading nodes.

\begin{lemma}[Leading Lemma]
  \label{lemma:leading}
  Suppose $w \in V$ is a leading node at time $t$. Then $\der{t}L_w(t)=\der{t}H_w(t)\in [1,1+\rho]$.
\end{lemma}
\begin{proof}
By \textbf{I4}, the claim follows if $w$ satisfies the slow condition at time $t$. As $w$ is a leading node at time $t$, there are $s \in \N$ and $v \in V$ satisfying
\begin{equation*}
\Psi_v^s(t) = L_w(t) - L_v(t) - 2s \kappa d(v,w) > 0.
\end{equation*}
In particular, $L_w(t)>L_v(t)$, so $w\neq v$. For any $y\in V$, we have
\begin{align*}
  L_w(t) - L_v(t) - 2 s \kappa d(v,w) &= \Psi_v^s(t)\\
  &\geq L_y(t) - L_v(t) - 2 s \kappa d(y,w).
\end{align*}
Rearranging this expression yields
\begin{equation*}
L_w(t) - L_y(t) \geq 2 s \kappa (d(v,w)-d(y,w)).
\end{equation*}
In particular, for any $y \in N_v$, $d(v,w) \geq d(y,w) - 1$ and hence
\begin{equation*}
L_y(t) - L_w(t) \leq 2s \kappa,
\end{equation*}
i.e., SC2 holds for $s$ at $w$.

Now consider $x \in N_v$ so that $d(x,w) = d(v,w) - 1$. Such a node exists because $v \neq w$. We obtain
\begin{equation*}
L_w(t) - L_y(t) \geq 2 s \kappa.
\end{equation*}
Thus SC1 is satisfied for $s$, i.e., indeed the slow condition holds at $w$ at time $t$.
\end{proof}

Lemma~\ref{lemma:leading} can readily be translated into a bound on the growth of $\Psi_w^s$ whenever $\Psi_w^s > 0$.

\begin{lemma}[Wait-up Lemma]
  \label{lemma:wait_up}
  Suppose $w \in V$ satisfies $\Psi_w^s(t) > 0$ for all $t \in (t_0, t_1]$. Then
  \[
  \Psi_w^s(t_1) \leq \Psi_w^s(t_0) - (L_w(t_1) - L_w(t_0)) + (1 + \rho) (t_1 - t_0).
  \]
\end{lemma}
\begin{proof}
  Fix $w \in V$, $s \in \N$ and $(t_0, t_1]$ as in the hypothesis of the lemma. For $v \in V$ and $t \in (t_0, t_1]$, define the function $f_v(t) = L_v(t) - 2 s \kappa d(v,w)$. Observe that
  \[
  \max_{v \in V} \{f_v(t)\}-L_w(t) = \Psi_w^s(t)\,.
  \]
Moreover, for any $v$ satisfying $f_v(t) = L_w(t) + \Psi_w^s(t)$, we have $L_v(t) - L_w(t) - 2s \kappa d(v,w) = \Psi_w^s(t) > 0$. Thus, Lemma~\ref{lemma:leading} shows that $v$ is in slow mode at time $t$. As (we assume that) logical clocks are differentiable, so is $f_v$, and it follows that $\der{t}f_v(t) \leq 1 + \rho$ for any $v \in V$ and time $t \in (t_0,t_1]$ satisfying $f_v(t) = \max_{x\in V} \{f_x(t)\}$. By Lemma~\ref{lem:max-bound}, it follows that $\max_{v\in V}\{f_v(t)\}$ grows at most at
 rate $1 + \rho$:
\begin{equation*}
\max_{v \in V} \{f_v(t_1)\}\leq \max_{v \in V} \{f_v(t_0)\}+(1+\rho)(t_1-t_0)\,.
\end{equation*}

We conclude that
\begin{align*}
  \Psi_w^s(t_1) - \Psi_w^s(t_0) &= \max_{v \in V} \{f_v(t_1)\} - L_w(t_1)\\
  &\qquad - (\max_{v \in V} \{f_v(t_0)\} - L_w(t_0))\\
  &\leq (1+\rho)(t_1 - t_0)-(L_w(t_1) - L_w(t_0)),
\end{align*}
which can be rearranged into the desired result.
\end{proof}
\begin{corollary}\label{cor:wait_up}
For all $s\in \N$ and times $t_1\geq t_0$, $\Psi^s(t_1)\le \Psi^s(t_0) + \rho (t_1-t_0)$.
\end{corollary}
\begin{proof}
Choose $w\in V$ such that $\Psi^s(t_1)=\Psi_w^s(t_1)$. As $\Psi_w^s(t)\geq 0$ for all times $t$, nothing is to show if $\Psi^s(t_1)=0$. Let $t\in [t_0,t_1)$ be the supremum of times from $t'\in [t_0,t_1)$ with the property that $\Psi_w^s(t')=0$. Because $\Psi_w^s$ is continuous, $t\neq t_0$ implies that $\Psi_w^s(t)=0$. Hence, $\Psi_w^s(t)\le \Psi_w^s(t_0)$. By \textbf{I2} and \cref{lemma:wait_up}, we get that
\begin{align*}
\Psi^s(t_1)&=\Psi_w^s(t_1)\\
&\le \Psi_w^s(t) - (L_w(t_1) - L_w(t)) + (1 + \rho) (t_1 - t)\\
&\le \Psi_w^s(t)+\rho(t_1-t)\\
&\le \Psi_w^s(t_0)+\rho(t_1-t_0)\\
&\le \Psi^s(t_0)+\rho(t_1-t_0).\qedhere
\end{align*}
\end{proof}

\subsection*{Trailing Nodes}

As $L_w(t_1) - L_w(t_0) \geq t_1-t_0$ at all times by \textbf{I2}, Lemma~\ref{lemma:catch_up} implies that $\Psi^s$ cannot grow faster than at rate $\rho$ when $\Psi^s(t) > 0$. This means that nodes whose clocks are far behind leading nodes can catch up, so long as the lagging nodes satisfy the fast condition and thus run at rate at least $1+\mu$ by \textbf{I3}. Our next task is to show that ``trailing nodes'' always satisfy the fast condition so that they are never too far behind leading nodes. The approach to showing this is similar to the one for Lemma~\ref{lemma:wait_up}, where now we need to exploit the fast condition.

\begin{definition}[$\Xi$ and Trailing Nodes]
For each $v\in V$, $s \in \N$, and $t \in \nnR$, we define
  \begin{equation*}
    \Xi_v^s(t) = \max_{w\in V} \{L_v(t) - L_w(t) -  (2 s + 1) \kappa d(v,w)\},
  \end{equation*}
where $d(v,w)$ denotes the distance between $v$ and $w$ in $G$. Moreover, set
\begin{equation*}
\Xi^s(t) = \max_{v\in V} \{\Xi_v^s(t)\}.
\end{equation*}
Finally, we say that $w\in V$ is a \emph{trailing node} at time $t$, if there is some $v\in V$ satisfying
\begin{equation*}
\Xi_v^s(t) = L_v(t) - L_w(t) - (2 s + 1) \kappa d(v,w) > 0.
\end{equation*}
\end{definition}

\begin{lemma}[Trailing Lemma]
  \label{lemma:trailing}
  If $w \in V$ is a trailing node at time $t$, then $\der{t}L_w(t)=(1+\mu)\der{t}H_w(t)\in [1+\mu,(1+\rho)(1+\mu)]$.
\end{lemma}
\begin{proof}
By \textbf{I3}, it suffices to show that $w$ satisfies the fast condition at time $t$. Let $s$ and $v$ satisfy
  \begin{align*}
    &L_v(t)-L_w(t) - (2s + 1) \kappa d(v,w)\\
    &\qquad = \max_{x\in V} \{L_v(t) - L_x(t) - (2 s + 1) \kappa d(v,x)\} > 0.
  \end{align*}
  In particular, $L_v(t)>L_w(t)$, implying that $v \neq w$. For $y \in V$, we have
  \begin{align*}
    &L_v(t) - L_w(t) - (2s + 1) \kappa d(v,w)\\
    &\qquad \geq L_v(t) - L_y(t) - (2s + 1) \kappa d(v,y).
  \end{align*}
  Thus for all neighbors $y\in N_w$,
  \begin{equation*}
    L_y(t) - L_w(t) + (2s + 1) \kappa (d(v,y) - d(v,w)) \geq 0.
  \end{equation*}
  It follows that
  \begin{equation*}
    \forall y \in N_v \colon L_w(t) - L_y(t) \leq (2s + 1)\kappa,
  \end{equation*}
  i.e., FC2 holds for $s$. As $v\neq w$, there is some node $x \in N_v$ with $d(v,x) = d(v,w) - 1$. Thus we obtain
  \begin{equation*}
    \exists x \in N_v\colon L_y(t) - L_w(t) \geq (2s + 1)\kappa,
  \end{equation*}
  showing FC1 for $s$, i.e., indeed the fast condition holds at $w$ at time $t$.
\end{proof}

Using Lemma~\ref{lemma:trailing}, we can show that if $\Psi^s_w(t_0)>0$, $w$ will eventually catch up. How long this takes can be expressed in terms of $\Psi^{s-1}(t_0)$, or, if $s=0$, $\calG$.

\begin{lemma}[Catch-up Lemma]\label{lemma:catch_up}
  Let $s\in \N$ and $v,w\in V$. Let $t_0$ and $t_1$ be times satisfying that
  \begin{equation*}
    t_1 \geq t_0 + \frac{\Xi_v^s(t_0)}{\mu}.
  \end{equation*}
  Then
  \begin{equation*}
    L_w(t_1) \geq t_1 - t_0 + L_v(t_0)-(2s+1)\kappa d(v,w).
  \end{equation*}
\end{lemma}
\begin{proof}
  W.l.o.g., we may assume that $t_1=t_0 + \Xi_v^s(t_0)/\mu$, as \textbf{I2} ensures that $\der{t}L_w(t)\geq 1$ at all times, i.e., the general statement readily follows. For any $x \in V$, define
  \[
  f_x(t) = t - t_0 + L_v(t_0) - L_x(t) - (2s + 1) \kappa d(v,x).
  \]
  Again by \textbf{I2}, it thus suffices to show that $f_w(t)\le 0$ for some $t\in [t_0, t_1]$.
  
  Observe that $\Xi_v^s(t_0)=\max_{x\in V} \{f_x(t_0)\}$. Thus, it suffices to show that $\max_{x\in V} \{f_x(t)\}$ decreases at rate $\mu$ so long as it is positive, as then $f_w(t_1)\le \max_{x\in V} \{f_x(t_1)\} \le 0$. To this end, consider any time $t\in [t_0,t_1]$ satisfying $\max_{x \in V} \{f_x(t)\} > 0$ and let $y \in V$ be any node such that $\max_{x\in V} \{f_x(t)\} = f_y(t)$. Then $y$ is trailing, as
  \begin{align*}
    \Xi_v^s(t)&=\max_{x\in V}\{L_v(t)-L_x(t)-(2s+1)\kappa d(v,x)\}\\
    &= L_v(t)-L_v(t_0)-(t-t_0)+\max_{x\in V}\{f_x(t)\}\\
    &= L_v(t)-L_v(t_0)-(t-t_0)+f_y(t)\\
    &= L_v(t)-L_y(t)-(2s+1)\kappa d(v,y)
  \end{align*}
  and
  \begin{align*}
    \Xi_v^s(t)&=L_v(t) - L_v(t_0) - (t - t_0) + \max_{x\in V} \{f_x(t)\}\\
    &> L_v(t) - L_v(t_0) - (t-t_0) \geq 0.
  \end{align*}
  Thus, by Lemma~\ref{lemma:trailing} we have that $\der{t}L_y(t)\geq 1+\mu$, implying $\der{t}f_y(t) = 1 - \der{t} L_y(t) \leq -\mu$.

  To complete the proof, assume towards a contradiction that $\max_{x\in V} \{f_x(t)\} > 0$ for all $t \in [t_0,t_1]$. Then, applying Lemma~\ref{lem:max-bound} again, we conclude that
  \begin{align*}
    \Xi_v^s(t_0)&=\max_{x\in V} \{f_x(t_0)\}\\
    &> -(\max_{x\in V} \{f_x(t_1)\} - \max_{x\in V} \{f_x(t_0)\})\\
    &\geq \mu(t_1-t_0) = \Xi_v^s(t_0),
  \end{align*}
  i.e., it must hold that $f_w(t)\le \max_{x\in V} \{f_x(t)\}\le 0$ for some $t \in [t_0,t_1]$.
\end{proof}

\subsection{Base Case and Global Skew}

We now prove that if $\Psi^s(0)$ is bounded for some $s\in \N$, it cannot grow significantly and thus remains bounded. This will both serve as an induction anchor for establishing our bound on the local skew and for bounding the global skew, as $\Psi^0(t)=\calG(t)$. In addition, we will deduce that even if the initial global skew $\calG(0)$ is large, at times $t\geq \calG(0)/\mu$, $\calG(t)$ is bounded by $\Gmax = (1-2\rho/\mu)\kappa D$.

To this end, we will apply \cref{lemma:catch_up} in the following form.
\begin{corollary}\label{cor:catch_up_same}
Let $s\in \N$ and $t_0$, $t_1$ be times satisfying
  \begin{equation*}
  t_1 \geq t_0 + \frac{\Psi^s(t_0)}{\mu}.
  \end{equation*}
  Then, for any $w \in V$ we have
  \begin{equation*}
    L_w(t_1) - L_w(t_0) \geq t_1 - t_0 +
      \Psi_w^s(t_0) - \kappa \cdot D.
  \end{equation*}
\end{corollary}
\begin{proof}
If $\Psi_w^s(t_0) - \kappa \cdot D\leq 0$, the claim is trivially satisfied due to \textbf{I2} guaranteeing that $\der{t}L_w(t)\ge 1$ at all times $t$. Hence, assume that $\Psi_w^s(t_0) - \kappa \cdot D>0$ and choose any $v$ so that
\begin{equation*}
\Psi_w^s(t_0) = L_v(t) - L_w(t) - 2s\kappa d(v,w).
\end{equation*}
We have that
\begin{align*}
\Xi_v^s(t_0)&\ge L_v(t) - L_w(t) - (2s+1)\kappa d(v,w)\\
&\ge L_v(t) - L_w(t) - 2s\kappa d(v,w)-\kappa\cdot D\\
&=\Psi_w^s(t_0) - \kappa \cdot D.
\end{align*}
As trivially $\Psi^s(t_0)\ge \Xi^s(t_0)\ge \Xi_v^s(t_0)$, we have that $t_1\ge t_0 + \Xi_v^s(t_0)/\mu$ and the claim follows by applying \cref{lemma:catch_up}.
\end{proof}

Combining this corollary with \cref{lemma:wait_up}, we can bound $\Psi^s$ at all times.
\begin{lemma}\label{lemma:psi_bound_same}
Fix $s\in \N$. If $\Psi^s(0)\le \kappa \cdot D / (1-\rho^2/\mu^2)$, then
\begin{equation*}
\Psi^s(t)\le \frac{\mu}{\mu-\rho}\cdot \kappa \cdot D.
\end{equation*}
at all times $t$. Otherwise,
\begin{equation*}
\Psi^s(t)\le \begin{cases}
\left(1+\frac{\rho}{\mu}\right)\cdot\Psi^s(0) & \mbox{if }t\le \frac{\Psi^s(0)}{\mu}\\
\kappa \cdot D + \frac{\rho}{\mu}\cdot  \left(1+\frac{\rho}{\mu}\right)\cdot\Psi^s(0)
 & \mbox{else.}
\end{cases}
\end{equation*}
\end{lemma}
\begin{proof}
For $t\le \Psi^s(0)/\mu$, the claim follows immediately from \cref{cor:wait_up} (and possibly using that $\Psi^s(0)\le \kappa \cdot D / (1-\rho^2/\mu^2)$). Concerning larger times, denote by $B$ the bound that needs to be shown and suppose that $\Psi^s(t_1)=B+\varepsilon$ for some $\varepsilon>0$ and minimal $t_1>\Psi^s(0)/\mu$. Choose $w\in V$ so that $\Psi_w^s(t_1)=\Psi^s(t_1)$ and $t_0$ such that $t_1 = t_0 + \Psi^s(t_0)/\mu$. Such a time must exist, because the function $f(t)=t_1-t-\Psi^s(t)/\mu$ is continuous and satisfies
\begin{align*}
f(t_1)=-\frac{\Psi^s(t_1)}{\mu}<0<t_1-\frac{\Psi^s(0)}{\mu}=f(t_0).
\end{align*}
We apply \cref{lemma:wait_up,cor:catch_up_same}, showing that
\begin{align*}
\Psi_w^s(t_1) &\le \Psi_w^s(t_0) - (L_w(t_1) - L_w(t_0)) + (1 + \rho) (t_1 - t_0)\\
&\le \kappa \cdot D + \rho (t_1-t_0)\\
&= \kappa\cdot D + \frac{\rho}{\mu}\Psi^s(t_0).
\end{align*}
We distinguish two cases. If $\Psi^s(0)\le \kappa \cdot D / (1-\rho^2/\mu^2)$, we have that
\begin{equation*}
\Psi^s(t_0) < \frac{\mu}{\mu-\rho}\cdot \kappa \cdot D +\varepsilon,
\end{equation*}
because $t_0<t_1$, leading to the contradiction
\begin{equation*}
\frac{\mu}{\mu-\rho}\cdot \kappa \cdot D + \varepsilon = \Psi^s(t_1)
< \left(1+\frac{\rho}{\mu-\rho}\right)\cdot \kappa \cdot D + \varepsilon.
\end{equation*}
On the other hand, if $\Psi^s(0)> \kappa \cdot D / (1-\rho^2/\mu^2)$, this is equivalent to
\begin{equation*}
\kappa \cdot D + \frac{\rho}{\mu}\cdot  \left(1+\frac{\rho}{\mu}\right)\cdot\Psi^s(0)
>\left(1+\frac{\rho}{\mu}\right)\cdot\Psi^s(0).
\end{equation*} 
Hence,
\begin{equation*}
\Psi^s(t_0) < \kappa \cdot D + \frac{\rho}{\mu}\cdot  \left(1+\frac{\rho}{\mu}\right)\cdot\Psi^s(0)+\varepsilon
\end{equation*}
and we get that
\begin{align*}
&\kappa \cdot D + \frac{\rho}{\mu}\cdot  \left(1+\frac{\rho}{\mu}\right)\cdot\Psi^s(0)+\varepsilon\\
 &\qquad= \Psi^s(t_1)\\
&\qquad< \kappa \cdot D + \frac{\rho}{\mu}\cdot \left(\kappa \cdot D + \frac{\rho}{\mu}\cdot  \left(1+\frac{\rho}{\mu}\right)\cdot\Psi^s(0) + \varepsilon\right).
\end{align*}
This implies the contradiction
\begin{equation*}
\Psi^s(0)<\frac{\kappa \cdot D}{1+\rho/\mu}+\frac{\rho}{\mu}\cdot \Psi^s(0)
\end{equation*}
to $\Psi^s(0)>\kappa \cdot D/(1-\rho^2/\mu^2)$.
\end{proof}

\begin{corollary}\label{cor:convergence}
Abbreviate $q=\frac{\rho}{\mu}\cdot \left(1+\frac{\rho}{\mu}\right)$ and assume that $q\le \frac{3}{4}$. For $i,s\in \N$ and times $t\ge 4(\Psi^s(0)+i\cdot\kappa\cdot D)/\mu$, it holds that
\begin{equation*}
\Psi^s(t)\le \frac{\kappa D}{1-q}+q^i\left(1+\frac{\rho}{\mu}\right)\Psi^s(0).
\end{equation*}
\end{corollary}
\begin{proof}
Consider the series given by $x_0=(1+\rho/\mu)\Psi^s_0$, $x_{i+1}=\kappa \cdot D + q x_i$, $t_0=0$, and $t_{i+1}=t_i+\frac{x_i}{\mu}$. By applying Lemma~\ref{lemma:psi_bound_same} with time $0$ replaced by time $t_i$ (i.e., shifting time) and $\Psi^s(0)$ by $x_i$, we can conclude that $x_i$ upper bounds $\Psi^s(t)$ at times $t\ge t_i$. Simple calculations show that $x_i\le \frac{\kappa D}{1-q}+q^i\Psi^s(0)$ and $t_i\le 4(\Psi^s(0)+i\cdot\kappa\cdot D)/\mu$, so the claim follows.
\end{proof}
In particular, $\Psi^s$ becomes bounded by $(1+O(\rho/\mu))\kappa D$ within $O(\Psi^s(0)/\mu)$ time. Plugging in $s=0$, we obtain a bound on the global skew.
\begin{corollary}\label{cor:global}
If $\frac{\rho}{\mu}\cdot \left(1+\frac{\rho}{\mu}\right)\le \frac{3}{4}$, it holds that
\begin{equation*}
\calG(t)\le \frac{\kappa D}{1-q}+q^i\left(1+\frac{\rho}{\mu}\right)\calG(0)
\end{equation*}
at all times $t\ge 4(\calG(0)+i\cdot\kappa\cdot D)/\mu$.
\end{corollary}
\begin{proof}
By applying \cref{cor:convergence} for $s=0$, noting that $\calG(t)=\Psi^0(t)$.
\end{proof}

\subsection{Bounding the Local Skew}
\label{sec:local-skew-bound}

In order to bound the local skew, we analyze the \emph{average} skew over paths in $G$ of various lengths. For long paths of $\Omega(D)$ hops, we will simply exploit that we already bounded the global skew, i.e., the skew between \emph{any} pair of nodes. For successively shorter paths, we inductively show that the average skew between endpoints cannot increase too quickly: reducing the length of a path by factor $\sigma$ can only increase the skew between endpoints by an additive constant term. Thus, paths of constant length (in particular edges) can only have a(n average) skew that is logarithmic in the network diameter.

In order to bound $\Psi^s$ in terms of $\Psi^{s-1}$, we need to apply the catch-up lemma in a different form.

\begin{corollary}\label{cor:catch_up_different}
  Let $s\in \Zpos$ and $t_0$, $t_1$ be times satisfying
  \begin{equation*}
  t_1 \geq t_0 + \frac{\Psi^{s-1}(t_0)}{\mu}.
  \end{equation*}
  Then, for any $w \in V$ we have
  \begin{equation*}
    L_w(t_1) - L_w(t_0) \ge t_1 - t_0 + \Psi_w^s(t_0).
  \end{equation*}
\end{corollary}
\begin{proof}
We have that $\Psi^{s-1}(t_0)\ge \Xi^{s-1}(t_0)$ and there is some $v\in V$ satisfying
\begin{align*}
\Psi_w^s(t_0)&=L_v(t_0)-L_w(t_0)-2s\kappa d(v,w).
\end{align*}
We apply \cref{lemma:catch_up} to $t_0$, $t_1$, $v$, $w$ and level $s-1$, yielding that
\begin{align*}
&L_w(t_1)-L_w(t_0)\\
&\qquad\ge t_1-t_0+L_v(t_0)-L_w(t_0)-(2s-1)\kappa d(v,w)\\
&\qquad\ge t_1-t_0+L_v(t_0)-L_w(t_0)-2s\kappa d(v,w)\\
&\qquad=t_1-t_0+\Psi_w^s(t_0).\qedhere
\end{align*}
\end{proof}

Combining this corollary with \cref{lemma:wait_up}, we can bound $\Psi^s$ at all times.
\begin{lemma}\label{lemma:psi_bound_different}
Fix $s\in \Zpos$ and suppose that $\Psi^{s-1}(t)\le \psi^{s-1}$ for all times $t$. Then
\begin{equation*}
\Psi^s(t)\le
\begin{cases}
\Psi^s(0)+\frac{\rho}{\mu}\cdot \psi^{s-1} & \mbox{if }t\le \frac{\psi^{s-1}}{\mu}\\
\frac{\rho}{\mu}\cdot \psi^{s-1} & \mbox{else.}
\end{cases}
\end{equation*}
\end{lemma}
\begin{proof}
For $t\le \psi^{s-1}/\mu$, the claim follows immediately from \cref{cor:wait_up}. To show the claim for $t>\psi^{s-1}/\mu$, assume for contradiction that it does not hold true and let $t_1$ be minimal such that there $\Psi^s(t_1)>\rho \psi^{s-1}/\mu + \varepsilon$ for some $\varepsilon>0$. Thus, there is some $w\in V$ so that
\begin{equation*}
\Psi_w^s(t_1)=\Psi^s(t_1)=\frac{\rho}{\mu}\cdot\psi^{s-1} + \varepsilon.
\end{equation*}
Applying \cref{cor:catch_up_different} with $t_0=t_1-\psi^{s-1}/\mu$ together with \cref{lemma:wait_up} yields the contradiction
\begin{align*}
\Psi^s_w(t_1)&\le \Psi^s_w(t_0)-(L_w(t_1)-L_w(t_0))+(1+\rho)(t_1-t_0)\\
&\le \rho (t_1-t_0)\\
& = \frac{\rho}{\mu}\cdot\psi^{s-1}.\qedhere
\end{align*}
\end{proof}

\begin{corollary}\label{cor:local}
Fix $s\in \N$. Suppose that $\Psi^s(t)\le \psi^s$ for all times $t$ and that $\calL(0)\le 2(s+1)\kappa$. Then
\begin{equation*}
\Psi^{s'}(t)\le \left(\frac{\rho}{\mu}\right)^{s'-s}\psi^s
\end{equation*}
for all $s'\geq s$ and times $t$.
\end{corollary}
\begin{proof}
Observe that $\calL(0)\le 2(s+1)\kappa$ implies that $\Psi^{s'}(0)=0$ for all $s'>s$. Thus, the statement follows from \cref{lemma:psi_bound_different} by induction on $s'$, where $\psi^{s'}=\rho \cdot \psi^{s'-1}/\mu$ and the base case is $s'=s$.
\end{proof}

\begin{corollary}\label{cor:local_stab}
Fix $s\in \N$. Suppose that $\Psi^s(t)\le \psi^s$ for all times $t$. Then
\begin{equation*}
\Psi^{s'}(t)\le \left(\frac{\rho}{\mu}\right)^{s'-s}\psi^s
\end{equation*}
for all $s'\geq s$ and times $t\ge \psi^s/(\mu-\rho)$.
\end{corollary}
\begin{proof}
Consider the times 
\begin{equation*}
t_{s'}=\sum_{i=1}^{s'-s}\left(\frac{\rho}{\mu}\right)^i\cdot \frac{\psi^s}{\mu}\le \frac{\psi^s}{\mu}\cdot \frac{1}{1-\rho/\mu}=\frac{\psi^s}{\mu-\rho}.
\end{equation*}
We apply \cref{lemma:psi_bound_different} inductively, where in step $s'>s$ we shift times by $-t_{s'}$. Thus, all considered times fall under the second case of \cref{lemma:psi_bound_different}, i.e., the initial values $\Psi^{s'}(0)$ (or rather $\Psi^{s'}(t_{s'})$) do not matter.
\end{proof}

\subsection{Putting Things Together}

It remains to combine the results on global and local skew to derive bounds that depend on the system parameters and initialization conditions only. First, we state the bounds on global and local skew that hold at all times. We emphasize that this bound on the local skew also bounds up to which level $s\in \N$ the algorithm needs to check $\FT1$ and $\FT2$, as larger local skews are impossible.
\begin{theorem}\label{thm:bound_levels}
Suppose that $\calL(0)\le (2s+1)\kappa$ for some $s\in \N$. Then
\begin{align*}
\calG(t)&\le \left(2s+\frac{\mu}{\mu-\rho}\right)\kappa D\\
\mbox{and}\qquad \calL(t)&\le \left(2s+\left\lceil \log_{\mu/\rho}\frac{\mu D}{\mu-\rho}\right\rceil+1\right)\kappa
\end{align*}
for all $t\in \nnR$.
\end{theorem}
\begin{proof}
As $\calL(0)\le (2s+1)\kappa$, we have that
\begin{equation*}
\Psi^s(0)\le \max_{v,w\in V}\{d(v,w)\}\cdot\kappa = \kappa \cdot D.
\end{equation*}
By \cref{lemma:psi_bound_same}, hence $\Psi^s(t)\le \frac{\mu}{\mu-\rho}\cdot \kappa \cdot D$ at all times $t$. Thus,
\begin{align*}
L_v(t)-L_w(t)-2s\kappa D &\le L_v(t)-L_w(t)-2s\kappa d(v,w)\\
&\le \Psi^s(t)\\
&\le \frac{\mu}{\mu-\rho}\cdot \kappa \cdot D
\end{align*}
for all $v,w\in V$ and times $t$, implying the stated bound on the global skew.

Concerning the local skew, apply \cref{cor:local} with $\psi^s=\frac{\mu}{\mu-\rho}\cdot \kappa \cdot D$ and $s'=s+\left\lceil \log_{\mu/\rho}\frac{\mu D}{\mu-\rho}\right\rceil$, yielding that
\begin{equation*}
\Psi^{s'}(t)\le \left(\frac{\rho}{\mu}\right)^{\lceil\log_{\mu/\rho}(\psi^s/\kappa)\rceil}\psi^s\le \kappa.
\end{equation*}
Hence, for all neighbors $v,w\in V$ and all times $t$,
\begin{align*}
L_v(t)-L_w(t)-2s'\kappa &= L_v(t)-L_w(t)-2s'\kappa d(v,w)\\
&\le \Psi^{s'}(t)\le \kappa,
\end{align*}
implying the claimed bound on the local skew.
\end{proof}
\cref{thm:bound_levels} bounds the number of levels $s\in \N$ for which the algorithm needs to check $\FT1$ and $\FT2$, depending on the local skew at initialization. It also shows that, if the system can be initialized with local skew at most $\kappa$, the system maintains the strongest bounds the algorithm guarantees at all times.
\begin{corollary}\label{cor:perfect_init}
Suppose that $\calL(0)\le \kappa$. Then
\begin{align*}
\calG(t)&\le \frac{\mu}{\mu-\rho}\cdot \kappa D\\
\mbox{and}\qquad \calL(t)&\le \left(\left\lceil \log_{\mu/\rho}\frac{\mu D}{\mu-\rho}\right\rceil+1\right)\kappa
\end{align*}
for all $t\in \nnR$.
\end{corollary}
If such highly accurate intialization is not possible, the algorithm will converge to the bounds from \cref{cor:perfect_init}.
\begin{theorem}\label{thm:stable_skew}
Suppose that $\mu>2\rho$. Then there is some $T\in O\left(\frac{\calG(0)+\kappa D}{\mu-2\rho}\right)$ such that
\begin{align*}
\calG(t)&\le \frac{\mu}{\mu-2\rho}\cdot \kappa D\\
\mbox{and}\qquad \calL(t)&\le \left(\left\lceil \log_{\mu/\rho}\frac{\mu D}{\mu-2\rho}\right\rceil+1\right)\kappa
\end{align*}
for all times $t\ge T$.
\end{theorem}
\begin{proof}
By assumption, 
\begin{equation*}
q=\frac{\rho}{\mu}\cdot \left(1+\frac{\rho}{\mu}\right)\le \frac{1}{2}\cdot \frac{3}{2}=\frac{3}{4}.
\end{equation*}
Fix some sufficiently small constant $\varepsilon>0$ such that
\begin{equation*}
\frac{\kappa D}{1-q}+\varepsilon \kappa D \le \frac{\kappa D}{1-2\rho/\mu};
\end{equation*}
since $q\le \frac{3}{2}\cdot\frac{\rho}{\mu}$, such a constant exists. Choose $i\in \N$ minimal with the property that
$q^i\left(1+\frac{\rho}{\mu}\right)\calG(0)\le \varepsilon \kappa D$. Therefore, by \cref{cor:global},
\begin{equation*}
\calG(t)\le \frac{\mu \kappa D}{\mu-2\rho}
\end{equation*}
at all times $t\ge 4(\calG(0)+i\kappa D)/\mu$. Noting that $\Psi^0(t)=\calG(t)$, analogously to \cref{thm:bound_levels} we can now apply \cref{cor:local_stab} to infer the desired bound on the local skew for times
\begin{equation*}
t\ge \frac{4(\calG(0)+i\kappa D)}{\mu}+\frac{\mu \kappa D}{(\mu-\rho)(\mu-2\rho)}.
\end{equation*}
Consequently, it remains to show that the right hand side of this inequality is indeed in $O\left(\frac{\calG(0)+\kappa D}{\mu-2\rho}\right)$. As $\mu-\rho\ge \mu/2$, this is immediate for the second term. Concerning the first term, our choice of $i$ and $q\le 3/4$ yield that $i\in O\left(\log \frac{\calG(0)}{\kappa D}\right)$. Because for $x\ge y>0$ it holds that $x\ge \log(x/y)\cdot y$, we can bound 
\begin{equation*}
\frac{4(\calG(0)+i\kappa D)}{\mu}\in O\left(\frac{\calG(0)+\kappa D}{\mu-2\rho}\right).\qedhere
\end{equation*}
\end{proof}
\cref{thm:gcs} is an immediate corollary of \cref{thm:stable_skew}.

\subsection{Proof of Lemma~\ref{lem:max-bound}}
\label{sec:max-bound}

\begin{proof}
  We prove the stronger claim that for all $a, b$ satisfying $t_0 \leq a < b \leq t_1$, we have
  \begin{equation}
      \frac{F(b) - F(a)}{b - a} \leq r. \label{eqn:avg-slope-bound}
  \end{equation}
  To this end, suppose to the contrary that there exist $a_0 < b_0$ satisfying $(F(b_0) - F(a_0)) / (b_0 - a_0) \geq r + \e$ for some $\e > 0$. We define a sequence of nested intervals $[a_0, b_0] \supset [a_1, b_1] \supset \cdots$ as follows. Given $[a_j, b_j]$, let $c_j = (b_j + a_j) / 2$ be the midpoint of $a_j$ and $b_j$. Observe that
  \begin{align*}
  \frac{F(b_j) - F(a_j)}{b_j - a_j} &= \frac 1 2 \frac{F(b_j) - F(c_j)}{b_j - c_j} + \frac{1}{2} \frac{F(c_j) - F(a_j)}{c_j - a_j}\\
  &\geq r + \e,
  \end{align*}
  so that
  \[
  \frac{F(b_j) - F(c_j)}{b_j - c_j} \geq r + \e \quad\text{or}\quad \frac{F(c_j) - F(a_j)}{c_j - a_j} \geq r + \e.
  \]
  If the first inequality holds, define $a_{j+1} = c_j$, $b_{j+1} = b_j$, and otherwise define $a_{j+1} = a_j$, $b_j = c_j$. From the construction of the sequence, it is clear that for all $j$ we have
  \begin{equation}
    \label{eqn:average-slope}
    \frac{F(b_j) - F(a_j)}{b_j - a_j} \geq r + \e.
  \end{equation}
  Observe that the sequences $\set{a_j}_{j = 0}^\infty$ and $\set{b_j}_{j=0}^\infty$ are both bounded and monotonic, hence convergent. Further, since $b_j - a_j = \frac{1}{2^j}(b_0 - a_0)$, the two sequences share the same limit.

  Define
  \[
  c = \lim_{j \to \infty} a_j = \lim_{j \to \infty} b_j,
  \]
  and let $f \in \mathcal{F}$ be a function satisfying $f(c) = F(c)$. By the hypothesis of the lemma, we have $f'(c) \leq r$, so that
  \[
  \lim_{h \to 0} \frac{f(c + h) - f(h)}{h} \leq r.
  \]
  Therefore, there exists some $h > 0$ such that for all $t \in [c - h, c + h]$, $t \neq c$, we have
  \[
  \frac{f(t) - f(c)}{t - c} \leq r + \frac 1 2 \e.
  \]
  Further, from the definition of $c$, there exists $N \in \N$ such that for all $j \geq N$, we have $a_j, b_j \in [c - h, c + h]$. In particular this implies that for all sufficiently large $j$, we have
  \begin{align}
    \frac{f(c) - f(a_j)}{c - a_j} &\leq r + \frac 1 2 \e, \label{eqn:aj-ineq}\\
    \frac{f(b_j) - f(c)}{b_j - c} &\leq r + \frac 1 2 \e. \label{eqn:bj-ineq}
  \end{align}
  Since $f(a_j) \leq F(a_j)$ and $f(c) = F(c)$,~(\ref{eqn:aj-ineq}) implies that for all $j \geq N$,
  \[
  \frac{F(c) - F(a_j)}{c - a_j} \leq r + \frac 1 2 \e.
  \]
  However, this expression combined with~(\ref{eqn:average-slope}) implies that for all $j \geq N$
  \begin{equation}
    \frac{F(b_j) - F(c)}{b_j - c} \geq r + \e. \label{eqn:bj-slope}
  \end{equation}
  Since $F(c) = f(c)$, the previous expression together with~(\ref{eqn:bj-ineq}) implies that for all $j \geq N$ we have $f(b_j) < F(b_j)$.

  For each $j \geq N$, let $g_j \in \mathcal{F}$ be a function such that $g_j(b_j) = F(b_j)$. Since $\mathcal{F}$ is finite, there exists some $g \in \mathcal{F}$ such that $g = g_j$ for infinitely many values $j$. Let $j_0 < j_1 < \cdots$ be the subsequence such that $g = g_{j_k}$ for all $k \in \N$. Then for all $j_k$, we have $F(b_{j_k}) = g(b_{j_k})$. Further, since $F$ and $g$ are continuous, we have
  \[
  g(c) = \lim_{k \to \infty} g(b_{j_k}) = \lim_{k \to \infty} F(b_{j_k}) = F(c) = f(c).
  \]
  By~(\ref{eqn:bj-slope}), we therefore have that for all $k$
  \[
  \frac{g(b_{j_k}) - g(c)}{b_{j_k} - c} = \frac{F(b_j) - F(c)}{b_j - c} \geq r + \e.
  \]
  However, this final expression contradicts the assumption that $g'(c) \leq r$. Therefore,~(\ref{eqn:avg-slope-bound}) holds, as desired.
\end{proof}

%% file: paper.bbl
\begin{thebibliography}{10}
\providecommand{\url}[1]{#1}
\csname url@samestyle\endcsname
\providecommand{\newblock}{\relax}
\providecommand{\bibinfo}[2]{#2}
\providecommand{\BIBentrySTDinterwordspacing}{\spaceskip=0pt\relax}
\providecommand{\BIBentryALTinterwordstretchfactor}{4}
\providecommand{\BIBentryALTinterwordspacing}{\spaceskip=\fontdimen2\font plus
\BIBentryALTinterwordstretchfactor\fontdimen3\font minus
  \fontdimen4\font\relax}
\providecommand{\BIBforeignlanguage}[2]{{%
\expandafter\ifx\csname l@#1\endcsname\relax
\typeout{** WARNING: IEEEtran.bst: No hyphenation pattern has been}%
\typeout{** loaded for the language `#1'. Using the pattern for}%
\typeout{** the default language instead.}%
\else
\language=\csname l@#1\endcsname
\fi
#2}}
\providecommand{\BIBdecl}{\relax}
\BIBdecl

\bibitem{foster2015trends}
H.~D. Foster, ``Trends in functional verification: A 2014 industry study,'' in
  \emph{52nd Annual Design Automation Conference}.\hskip 1em plus 0.5em minus
  0.4em\relax ACM, 2015, p.~48.

\bibitem{martin1986compiling}
A.~J. Martin, ``Compiling communicating processes into delay-insensitive vlsi
  circuits,'' \emph{Dist.\ comp.}, vol.~1, no.~4, pp. 226--234, 1986.

\bibitem{martin1990limitations}
------, ``The limitations to delay-insensitivity in asynchronous circuits,'' in
  \emph{Beauty is our business}.\hskip 1em plus 0.5em minus 0.4em\relax
  Springer, 1990, pp. 302--311.

\bibitem{manohar2017eventual}
R.~Manohar and Y.~Moses, ``The eventual c-element theorem for delay-insensitive
  asynchronous circuits,'' in \emph{23rd IEEE International Symposium on
  Asynchronous Circuits and Systems}.\hskip 1em plus 0.5em minus 0.4em\relax
  IEEE, 2017, pp. 102--109.

\bibitem{manohar2019asynchronous}
------, ``Asynchronous signalling processes,'' in \emph{25th IEEE Int.\
  Symposium on Asynchronous Circuits and Systems}.\hskip 1em plus 0.5em minus
  0.4em\relax IEEE, 2019, pp. 68--75.

\bibitem{chapiro1984globally}
D.~M. Chapiro, ``Globally-asynchronous locally-synchronous systems.'' Stanford
  Univ CA Dept of Computer Science, Tech. Rep., 1984.

\bibitem{teehan2007survey}
P.~Teehan, M.~Greenstreet, and G.~Lemieux, ``A survey and taxonomy of gals
  design styles,'' \emph{IEEE Design \& Test of Computers}, vol.~24, no.~5, pp.
  418--428, 2007.

\bibitem{dobkin2004data}
R.~Dobkin, R.~Ginosar, and C.~P. Sotiriou, ``Data synchronization issues in
  gals socs,'' in \emph{10th International Symposium on Asynchronous Circuits
  and Systems}.\hskip 1em plus 0.5em minus 0.4em\relax IEEE, 2004, pp.
  170--179.

\bibitem{yun1996pausible}
K.~Y. Yun and R.~P. Donohue, ``Pausible clocking: A first step toward
  heterogeneous systems,'' in \emph{Proc.\ Int.\ Conference on Computer Design.
  VLSI in Computers and Processors}.\hskip 1em plus 0.5em minus 0.4em\relax
  IEEE, 1996, pp. 118--123.

\bibitem{fan2009analysis}
X.~Fan, M.~Krsti{\'c}, and E.~Grass, ``Analysis and optimization of pausible
  clocking based gals design,'' in \emph{IEEE International Conference on
  Computer Design}.\hskip 1em plus 0.5em minus 0.4em\relax IEEE, 2009, pp.
  358--365.

\bibitem{DDX95low}
L.~R. Dennison, W.~J. Dally, and D.~Xanthopoulos, ``Low-latency plesiochronous
  data retiming,'' in \emph{Proceedings Sixteenth Conference on Advanced
  Research in VLSI}.\hskip 1em plus 0.5em minus 0.4em\relax IEEE, 1995, pp.
  304--315.

\bibitem{CG03efficient}
A.~Chakraborty and M.~R. Greenstreet, ``Efficient self-timed interfaces for
  crossing clock domains,'' in \emph{9th International Symposium on
  Asynchronous Circuits and Systems}.\hskip 1em plus 0.5em minus 0.4em\relax
  IEEE, 2003, pp. 78--88.

\bibitem{lenzen10tight}
C.~Lenzen, T.~Locher, and R.~Wattenhofer, ``{Tight Bounds for Clock
  Synchronization},'' \emph{Journal of the ACM}, vol.~57, no.~2, pp. 1--42,
  2010.

\bibitem{Kuhn2009-gradient}
\BIBentryALTinterwordspacing
F.~Kuhn and R.~Oshman, ``{Gradient Clock Synchronization Using Reference
  Broadcasts},'' in \emph{Principles of Distributed Systems, 13th International
  Conference}, 2009, pp. 204--218. [Online]. Available:
  \url{https://doi.org/10.1007/978-3-642-10877-8\_17}
\BIBentrySTDinterwordspacing

\bibitem{martins2015open}
M.~Martins, J.~M. Matos, R.~P. Ribas, A.~Reis, G.~Schlinker, L.~Rech, and
  J.~Michelsen, ``Open cell library in 15nm freepdk technology,'' in
  \emph{Proceedings of the 2015 Symposium on International Symposium on
  Physical Design}.\hskip 1em plus 0.5em minus 0.4em\relax ACM, 2015, pp.
  171--178.

\bibitem{Kuhn2010-optimal}
\BIBentryALTinterwordspacing
F.~Kuhn, C.~Lenzen, T.~Locher, and R.~Oshman, ``Optimal gradient clock
  synchronization in dynamic networks,'' \emph{CoRR}, vol. abs/1005.2894, 2010.
  [Online]. Available: \url{http://arxiv.org/abs/1005.2894}
\BIBentrySTDinterwordspacing

\bibitem{m-gtmo-81}
L.~R. Marino, ``{General Theory of Metastable Operation},'' \emph{{IEEE}
  Transactions on Computers}, vol.~30, no.~2, pp. 107--115, 1981.

\bibitem{mota2000flexible}
M.~Mota, J.~Christiansen, S.~Debieux, V.~Ryjov, P.~Moreira, and A.~Marchioro,
  ``A flexible multi-channel high-resolution time-to-digital converter asic,''
  in \emph{IEEE Nuclear Science Symp.}, vol.~2, 2000, pp. 9--155.

\bibitem{fisher85synchronizing}
{Fisher} and {Kung}, ``{Synchronizing Large VLSI Processor Arrays},''
  \emph{IEEE Transactions on Computers}, vol. C-34, no.~8, pp. 734--740, 1985.

\bibitem{boksberger2003approximation}
P.~Boksberger, F.~Kuhn, and R.~Wattenhofer, ``On the approximation of the
  minimum maximum stretch tree problem,'' \emph{Technical report/ETH,
  Department of Computer Science}, vol. 409, 2003.

\bibitem{blog}
\BIBentryALTinterwordspacing
M.~James, ``Linear solver in linear time.'' [Online]. Available:
  \url{https://www.i-programmer.info/news/181-algorithms/5573-linear-solver-in-linear-time.html}
\BIBentrySTDinterwordspacing

\bibitem{fklw-fadcfa-18}
M.~F\"ugger, A.~Kinali, C.~Lenzen, and B.~Wiederhake, ``{Fast All-Digital Clock
  Frequency Adaptation Circuit for Voltage Droop Tolerance},'' in \emph{Symp.\
  on Asynchronous Circuits and Systems}, 2018.

\bibitem{Rudin1976-principles}
W.~Rudin, \emph{\BIBforeignlanguage{English}{Principles of {Mathematical}
  {Analysis}}}, 3rd~ed.\hskip 1em plus 0.5em minus 0.4em\relax New York:
  McGraw-Hill Education, 1976.

\end{thebibliography}
